\newif\ifconf

\ifconf
\documentclass[sigconf]{acmart}
\else
\documentclass[sigconf,nonacm]{acmart}
\fi

\usepackage{algorithm,algorithmicx}
\usepackage[noend]{algpseudocode}

\makeatletter
\def\thm@space@setup{%
  \thm@preskip=4.0pt plus 1.0pt minus 2.0pt 
  \thm@postskip=\thm@preskip
}
\makeatother

\usepackage{macros}

\usepackage{booktabs}
\usepackage[font=small,skip=0.1cm]{caption}
\usepackage{multirow, makecell}
\usepackage{subcaption}
\usepackage{tabularht}
\usepackage{tikz}
\usepackage{varwidth}
\usepackage{enumitem}

\usepackage[compact]{titlesec}
\titlespacing*{\section}{0pt}{3pt}{2pt}
\titlespacing*{\subsection}{0pt}{2pt}{1pt}

\usepackage[font=small,aboveskip=0pt,belowskip=0pt,tableposition=above]{caption}

\allowdisplaybreaks
\providecommand{\figfont}{\footnotesize} 

\makeatletter
\renewcommand\paragraph{\@startsection{paragraph}{4}{\parindent}%
  {0pt}%
  {-3.5\p@}%
  {\ACM@NRadjust{\@parfont\@adddotafter}}}
\makeatother

\newcommand{\aref}[1]{%
  \ifconf
    the full version of our paper~\cite{tseng2022parallel_arxiv}%
  \else
    \cref{#1}%
  \fi
}
\newcommand{\Aref}[1]{%
  \ifconf
    The full version of our paper~\cite{tseng2022parallel_arxiv}%
  \else
    \Cref{#1}%
  \fi
}
\makeatletter
\newcommand{\noconf}[1]{%
  \ifconf\else
    #1%
  \fi
}
\makeatother

\providecommand{\ab}{\allowbreak}
\providecommand{\smallabs}[1]{\lvert#1\rvert}

\providecommand{\tsc}{\textsc}
\providecommand{\combine}{\textsc{Combine}}
\providecommand{\rank}{\operatorname{rank}}
\providecommand{\rmin}{\operatorname{rmin}}
\providecommand{\rmax}{\operatorname{rmax}}

\providecommand{\Break}{\State \textbf{break}\ }
\providecommand{\Pfor}[1]{\For{#1} \textbf{in parallel}} 
\providecommand{\Local}{\State}
\makeatletter\algnewcommand{\LComment}[1]{\Statex \hskip\ALG@thistlm  \(\triangleright\) #1}\makeatother

\copyrightyear{2022}
\acmYear{2022}
\setcopyright{rightsretained}

\acmConference[SPAA '22]{Proceedings of the 34th ACM Symposium on Parallelism in Algorithms and Architectures}{July 11--14, 2022}{Philadelphia, PA, USA} \acmBooktitle{Proceedings of the 34th ACM Symposium on Parallelism in Algorithms and Architectures (SPAA '22), July 11--14, 2022, Philadelphia, PA, USA} \acmDOI{10.1145/3490148.3538584}
\acmISBN{978-1-4503-9146-7/22/07}

\usepackage{etoolbox}
\makeatletter
\patchcmd{\maketitle}{\@copyrightpermission}{
   \begin{minipage}{0.3\columnwidth}
     \href{https://creativecommons.org/licenses/by/4.0/}{\includegraphics[width=0.90\textwidth]{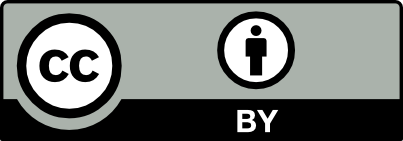}}
   \end{minipage}\hfill
   \begin{minipage}{0.7\columnwidth}
     \href{https://creativecommons.org/licenses/by/4.0/}{This work is licensed under a Creative Commons Attribution International 4.0 License.}
   \end{minipage}

   \vspace{5pt}
}{}{}

\makeatother

\begin{document}

\title{Parallel Batch-Dynamic Minimum Spanning Forest and the Efficiency of Dynamic Agglomerative Graph Clustering}
\ifconf\else
     \titlenote{This is the full version of the paper appearing in the ACM Symposium on Parallelism in Algorithms and Architectures (SPAA), 2022.}
\fi

\author{Tom Tseng}
\orcid{0000-0002-6422-288X}
\affiliation{
  \institution{MIT CSAIL}
  \city{Cambridge}
  \state{Massachusetts}
  \country{USA}
}
\email{tomtseng@mit.edu}

\author{Laxman Dhulipala}
\affiliation{
  \institution{University of Maryland}
  \city{College Park}
  \state{Maryland}
  \country{USA}
}
\email{laxman@umd.edu}

\author{Julian Shun}
\affiliation{
  \institution{MIT CSAIL}
  \city{Cambridge}
  \state{Massachusetts}
  \country{USA}
}
\email{jshun@mit.edu}


\begin{abstract}
  Hierarchical agglomerative clustering (HAC) is a popular algorithm for clustering data, but despite its importance, no dynamic
algorithms for HAC with good theoretical guarantees exist. In this paper, we study
dynamic HAC on edge-weighted graphs. As single-linkage HAC
reduces to computing a minimum spanning forest (MSF), our first result is
a parallel batch-dynamic algorithm for maintaining MSFs. 
On a batch of $k$ edge insertions or deletions, our batch-dynamic MSF algorithm
runs in $O(k\log^6 n)$ expected amortized work and $O(\log^4 n)$ span with high
probability. 
It is the first fully dynamic MSF algorithm handling batches of edge
updates with polylogarithmic work per update and polylogarithmic
span. Using our MSF algorithm, we obtain a parallel batch-dynamic algorithm that can answer queries about 
single-linkage graph HAC clusters.

Our second result is that dynamic graph HAC is significantly harder for other common
linkage functions. For example, assuming the strong exponential time hypothesis, 
dynamic graph HAC requires $\Omega(n^{1-o(1)})$ work per update or query on a graph with $n$
vertices for complete linkage, weighted average linkage, and average linkage. For complete
linkage and weighted average linkage, the bound still holds even for incremental or
decremental algorithms and even if we allow $\operatorname{poly}(n)$-approximation. For
average linkage, the bound weakens to $\Omega(n^{1/2 - o(1)})$ for incremental and
decremental algorithms, and the bounds still hold when allowing
$n^{o(1)}$-approximation.

\end{abstract}


\maketitle

\section{Introduction}\label{sec:intro}

Clustering is a fundamental technique in data mining and unsupervised learning
that organizes data into meaningful groups. In this paper, we study 
\emph{hierarchical agglomerative clustering} (HAC) algorithms. HAC constructs a
hierarchy of clusters over a set of points by starting with each point
in a separate cluster and merging the two most similar clusters until all points are merged. 
The similarity between clusters is specified by a \emph{linkage function}. Popular linkage
functions include single linkage, complete linkage, average linkage, 
and weighted average linkage, with average linkage perhaps being the most widely used.
Several popular clustering algorithms are based on single linkage as well~\cite{havens2009vat,bateni2017affinity}.
HAC on $n$ points can be solved in cubic work in general, and 
several common linkage functions require only quadratic work~\cite{benzecri1982construction}. Quadratic work is optimal
in the sense that if the input is an $n \times n$ similarity matrix for the $n$
points, then all matrix entries need to be read to compute HAC.

Because the similarity matrix has lots of negligible
entries in many scenarios, Dhulipala et al.~\cite{dhulipala2021hierarchical} recently studied \emph{graph-based HAC} (graph HAC) as opposed to the traditional \emph{point-based HAC}. In graph HAC,
not all similarities between points need to be
specified. Instead, the input is a graph with edges weighted by the similarity
between their endpoints. Dhulipala et al.\ develop exact and approximate algorithms for
graph HAC with subquadratic work on sparse graphs and empirically showed that the resulting clusters are of similar quality to those of point-based HAC.

Modern data sets are large and are often rapidly changing, and so a natural question
is whether we can compute HAC over a dynamic data set. Even with subquadratic work, it is
inefficient to statically re-compute HAC on every update of a large,
dynamically changing graph. Little research has been done on dynamic HAC. Graph HAC
seems more likely to yield fast dynamic algorithms than point-based HAC---a graph
update can be as granular as updating a similarity between one pair of vertices,
whereas updating points in point-based HAC incurs $\Omega(n)$ changes in the
similarity matrix. As such, this paper aims to study whether graph HAC
allows efficient dynamic algorithms under edge insertions and deletions.

The canonical output for HAC is a dendrogram showing the hierarchical clustering, but there are graphs for which
one edge update can completely change the structure of the dendrogram. 
It therefore seems that a dynamic HAC algorithm that explicitly maintains a dendrogram will
have poor worst-case update time. We hence examine dynamic graph HAC algorithms
with more restricted query outputs, e.g., queries of the form ``are query vertices $s$ and $t$
in the same cluster if we agglomeratively cluster until all similarities are below query threshold $\theta$.''

With this form of query, single-linkage graph HAC indeed admits efficient dynamic algorithms.
As single-linkage HAC reduces to
computing a minimum spanning forest (MSF)~\cite{gower1969minimum}, we can solve
dynamic single-linkage HAC by first applying a dynamic MSF algorithm. The state-of-the-art
dynamic sequential MSF algorithm achieves $O(\log^4 n / \log \log
n)$ amortized work per edge update to maintain an MSF~\cite{holm2015faster}. Then, storing the
MSF in a dynamic trees data structure~\cite{sleator1983data} allows us to answer the queries in logarithmic work.
To support a high velocity of updates, however, we may want a
\emph{batch-dynamic} algorithm that can batch together
updates and exploit parallelism across a batch. Though there are efficient
parallel batch-dynamic algorithms for connectivity and incremental MSF~\cite{acar2019parallel,anderson2020work}, 
no such algorithm has been developed for general dynamic MSF.

This discussion raises two questions: (1) Can we develop a
parallel batch-dynamic MSF algorithm, hence giving an parallel batch-dynamic algorithm for
single-linkage graph HAC? (2) Do other linkage functions
also admit dynamic algorithms with polylogarithmic work per update?

In this paper, we give a parallel batch-dynamic MSF algorithm achieving
$O(k\log^6 n)$ expected amortized work and $O(\log^3 n \log k)$ span with high
probability (w.h.p.)\footnote{We say that an event
occurs \emph{with high probability (w.h.p.)} if it occurs with probability at least
$1-1/n^c$ for any $c \ge 1$, where constants inside asymptotic bounds can
depend on $\poly{c}$.} for a batch of $k$ edge insertions or $k$
edge deletions.
Moreover, our MSF result is of independent interest outside the
context of clustering. Prior
to our algorithm, there was not even a batch-decremental MSF algorithm with 
polylogarithmic span achieving $O(kn)$ work on edge deletions.

We first give a
parallel batch-decremental MSF algorithm achieving 
$O(\log^4 n)$ expected amortized work per edge and
$O(\log^3 n \log k)$ span w.h.p.\ per batch. A key challenge in parallelizing the 
decremental MSF algorithm is fetching the $k$ lightest edges incident to a connected component in low span. We solve this approximately by augmenting an internal data structure with quantile summaries.
Then, we parallelize Holm et al.'s reduction of fully dynamic MSF to decremental MSF~\cite{holm2001poly} 
to obtain our batch-dynamic MSF algorithm.

\begin{table}\centering
  \figfont
  \begin{tabular}{@{}lllll@{}}
  \toprule
  \multirow{2}{*}{Problem} & \multicolumn{4}{c}{Work lower bounds} \\
                           & Preprocess & Update & Query & Conjecture \\
  \hline
  \multirow{4}{*}{\shortstack[l]{HAC (complete or\\ weighted average)}}
   & $\on{poly}(n)$ & $n^{1 - \eps}$ & $n^{1 - \eps}$ & SETH \\
   & $m^{1+\delta-\eps}$ & $m^{\delta - \eps}$ & $m^{2\delta - \eps}$ & Triangle \\
   & $m^{4/3-\eps}$ & $m^{\alpha - \eps}$ & $m^{2/3 - \alpha - \eps}$ & 3SUM \\
   & $\on{poly}(n)$ & $m^{1/2 - \eps}$ & $m^{1 - \eps}$ & OMv \\
   \hline
   \multirow{3}{*}{\shortstack[l]{\#HAC (complete or\\ weighted average)}}
   & $\on{poly}(n)$ & $n^{1 - \eps}$ & $n^{1 - \eps}$ & SETH \\
   & $m^{1+\delta-\eps}$ & $m^{\delta - \eps}$ & $m^{2\delta - \eps}$ & Triangle \\
   & $\on{poly}(n)$ & $m^{1/2 - \eps}$ & $m^{1 - \eps}$ & OMv \\
   \hline
  \multirow{2}{*}{HAC (average)} & $\on{poly}(n)$ & $n^{1-5c-\eps}$ & $n^{1-5c-\eps}$ & SETH \\
   & $\poly{n}$ & $n^{(1-c)/6 - \eps}$ & $n^{(1-c)/3 - \eps}$ & OMv \\
   \hline
  \multirowcell{2}[0pt][l]{Dec/Inc HAC\\(average)}
  & $\on{poly}(n)$ & $n^{(1-8c)/2-\eps}$ & $n^{(1-8c)/2-\eps}$ & SETH \\
  & & & &\\
   \hline
  \multirow{3}{*}{\#HAC (average)}
   & $\on{poly}(n)$ & $n^{1-c-\eps}$ & $n^{1-c-\eps}$ & SETH \\
   & $m^{(1+\delta)(1-c)-\eps}$ & $m^{\delta(1-c) - \eps}$ & $m^{2\delta(1-c) - \eps}$ & Triangle \\
   & $\poly{n}$ & $m^{(1-c)/2 - \eps}$ & $m^{1-c - \eps}$ & OMv \\
   \hline
   \multirowcell{2}[0pt][l]{Dec/Inc \#HAC\\(average)}
  & $\on{poly}(n)$ & $n^{1-2c-\eps}$ & $n^{1-2c-\eps}$ & SETH \\
  & & & &\\
  \bottomrule
  \end{tabular}
  \caption{The table states asymptotic work bounds for dynamic graph HAC such that the listed
    conjecture (defined in \cref{sec:avw-lower-bounds}) would be falsified. For problems listed as ``HAC'', queries answer whether
    two query vertices are in the same cluster after agglomeratively clustering
    up to a query similarity threshold, and in ``\#HAC'', queries (given a query similarity threshold) answer with the number of clusters. The bounds allow $O(n^c)$-approximation for a
    constant $c \ge 0$. In the table, $\eps > 0$ is an arbitrarily small
    constant, $\alpha \in [1/6, 1/3]$, and $\delta > 0$ is some constant for
    which triangle detection takes $\Omega(m^{1+\delta-o(1)})$ work. 
    The same bounds also hold for partially dynamic algorithms except for the 
    average-linkage bounds based on SETH; 
    we list SETH-based partially dynamic average-linkage bounds 
    separately as ``Dec/Inc.''  
    The bounds are amortized for fully dynamic algorithms and worst-case for partially dynamic algorithms.
  }
  \label{tab:lower-bounds}
\end{table}

On the other hand, even under our restricted query model for dynamic HAC, we show polynomial conditional 
lower bounds on the work
of dynamic graph HAC for complete linkage, weighted average linkage, and average linkage, 
even with $n^{o(1)}$-approximation and even when restricted
to incremental or decremental algorithms. 
\Cref{tab:lower-bounds} summarizes our lower bounds. 
Our bounds build on past work 
showing that several dynamic problems have lower bounds conditional on
conjectures like the strong exponential time hypothesis
(SETH)~\cite{calabro2009complexity} via reductions~\cite{abboud2014popular,henzinger2015unifying}.

Our contributions are summarized as follows:
\begin{itemize}
  \item We parallelize a relative-error quantile summary data structure
    (see \aref{app:quantile-proofs-par}) and use it to solve parallel batch-decremental
    MSF in $O(\log^3 n \log k)$ span w.h.p.\ per batch of $k$ edge deletions and 
    $O(\log^3 n \log (1 + n/\Delta)) \le O(\log^4 n)$ expected amortized work per edge 
    where $\Delta$ is the average batch size across deletion operations 
    (\cref{sec:dec-msf}).
  \item We use batch-decremental MSF to solve parallel batch-dynamic MSF (and
    hence also parallel batch-dynamic single-linkage graph HAC) in $O(k
    \log^6 n)$ expected amortized work and $O(\log^3 n \log k)$ span w.h.p.\ on a
    batch of $k$ edge insertions or edge deletions (\cref{sec:dyn-msf}).
    These are the first decremental and fully dynamic MSF algorithms 
    achieving polylogarithmic work per update and
    polylogarithmic span per batch.
  \item We prove polynomial conditional work lower bounds for dynamic and
    partially dynamic graph HAC with complete linkage, weighted average linkage, and average 
    linkage (\cref{sec:hac}). For example, assuming the SETH, dynamic HAC takes
    $\Omega(n^{1-o(1)})$ per update or query for all of these linkage functions,
    even with $n^{o(1)}$-approximation.
\end{itemize}

\section{Related work}

\paragraph{Graph HAC} We use the definition of graph HAC by Dhulipala et al., who give
algorithms solving static graph HAC on $m$ edges and $n$ vertices in $O(m \log n)$
expected work for weighted-average-linkage HAC, $\tilde{O}(n\sqrt{m})$ work for
average-linkage HAC, and $O(m\log^2 n)$ work for approximate average-linkage
HAC~\cite{dhulipala2021hierarchical}. Older papers have also studied graph HAC 
but with weaker theoretical
guarantees~\cite{karypis1999chameleon,franti2006fast}.
Another line of work has developed the theoretical foundations of HAC by studying the objective function that it optimizes~\cite{Dasgupta2016,hac-reward,moseley-wang}.

\paragraph{Dynamic HAC} There is no prior work on dynamic HAC with good approximation or running time guarantees. Menon et
al.\ give an online approximate algorithm for point-based dynamic HAC~\cite{menon2019online}. 
Their algorithm does not have rigorous bounds on approximation quality
or worst-case running time. Other online clustering
algorithms like Perch~\cite{kobren2017hierarchical} and Grinch~\cite{monath2019scalable}
neither compute the same output as HAC nor approximate HAC in a provably efficient way. 

\paragraph{HAC lower bounds} Point-based HAC in Euclidean space is at least
as hard as finding the closest pair of points. Karthik and Manurangsi show
that, assuming the SETH, closest-pair in dimension
$\omega(\on{polylog}(n))$ requires $\Omega(n^{2-o(1)})$ work, and
$(1 + o(1))$-approximate closest-pair in dimension $\omega(\log n)$ requires
$\Omega(n^{1.5-o(1))})$ work. These lower bounds do not apply to graph HAC.

\paragraph{Dynamic MSF} In this paper, we focus on edge
insertions and deletions. For sequential dynamic MSF, Holm et al.\ give an
algorithm with $O(\log^4 n)$ amortized work per edge update, which was later
improved to $O(\log^4 n / \log \log n)$ amortized work per
update~\cite{holm2001poly,holm2015faster}. The best worst-case bound is
$O(n^{o(1)})$ work per update~\cite{nanongkai2017dynamic, chuzhoy2020deterministic}.

For parallel batched edge updates, Anderson et al.\ give an incremental MSF algorithm
that handles $k$ edge insertions in $O(k\log(1 + n/k))$ expected work and
$O(\log^2 n)$ span w.h.p.~\cite{anderson2020work}. Other existing algorithms are deterministic but
have much higher work bounds. Pawagi and Kaiser give an algorithm handling
insertions in $O(kn)$ work and $O(\log n \log k)$ span and deletions in
$O(n^2(1+ \log^2 k / \log n))$ work and $O(\log n + \log^2 k)$
span~\cite{pawagi1993optimal}. Shen and Liang give an algorithm handling
insertions and deletions in $O(n^2)$ work and $O(\log n \log k)$
span~\cite{shen1993parallel}. Ferragina and Luccio give an algorithm handling
insertions in $O(n \log \log \log n \log(m/n))$ work and $O(\log n)$ span and $k
= O(n)$ deletions in $O(kn \log \log \log n \log(m/n))$ work and $O(\log n
\log(m/n))$ span~\cite{ferragina1994batch}.

For parallel single edge updates, Kopelowitz et al.\ give an algorithm running
in $O(\sqrt{n} \log n)$ work and $O(\log n)$ span per
update~\cite{kopelowitz2018improved}. There are also many algorithms for the
harder problem of dynamic vertex updates, all of which cost $\Omega(n)$ work per
update. We refer the reader to Das and Ferragina's survey for an overview of algorithms for vertex
updates as well as for edge updates~\cite{das1995parallel}.

\section{Preliminaries}

\paragraph{Graph HAC}
We denote a graph by $G=(V,E)$. Graphs are undirected and simple unless
noted otherwise. For weighted graphs, we denote a weight or similarity of an edge
$\set{x, y}$ either by writing $w(x, y)$ where $w: E \to \R$ is a weight
function or by placing weight $w \in \R$ in a tuple $(\set{u, v}, w)$.  We
often denote $n = \abs{V}$ and $m = \abs{E}$. In our asymptotic bounds, we
assume $m = \Omega(n)$.  We denote the neighbors of $v \in V$ as $N(v)$. We
write $\on{Cut}(X, Y)$ to denote the set of edges between two sets of vertices
$X$ and $Y$. 

In graph HAC, we are
given a weighted undirected graph and a \emph{linkage function}
specifying the similarities between clusters. Each vertex starts in its own
cluster, and we compute a hierarchical clustering by repeatedly merging the two
most similar clusters, i.e., the clusters connected by the highest-weight edge. 

In \emph{single linkage}, the similarity $\mathcal{W}(X, Y)$ between two
clusters $X$ and $Y$ is the maximum similarity between two vertices in $X$
and $Y$, i.e., $\max_{(x,y) \in \on{Cut}(X, Y)} w(x,y)$, whereas \emph{complete
linkage} takes the minimum similarity. In \emph{average linkage} (which still applies to 
weighted graphs), the similarity is $\sum_{(x,y) \in
\on{Cut}(X,Y)} w(x,y) / (\abs{X}\abs{Y})$. In \emph{weighted average linkage}, if a cluster $Z$ is formed by merging clusters $X$ and $Y$,
then the similarity between $Z$ and an adjacent cluster $U$ is $(\mathcal{W}(X,
U) + \mathcal{W}(Y, U)) / 2$ if edges $(X, U)$ and $(Y, U)$ both exist and is
otherwise the weight of the existing edge.

\paragraph{Parallel model.}
We use the work-span model with arbitrary forking, a closely related
model to the classic CRCW PRAM model~\cite{jaja1992introduction,blelloch1996programming}.
Running time bounds are in terms of \emph{work} and \emph{span} (depth). The work of an algorithm is
the total number of instructions, and the span
is the length of the longest chain of sequentially dependent instructions.
We assume that concurrent reads and writes are supported in $O(1)$ work and span.
Procedures can fork other procedure calls
to run in parallel and can wait for forked calls to finish. 

\paragraph{Parallel primitives.}
We use several existing parallel primitives in our algorithms.
Unordered sets can be stored in parallel dictionaries using linear space and handling batch
insertions or deletions of $k$ elements in $O(k)$ work and $O(\log^* k)$ span
w.h.p.~\cite{gil1991towards}. Lookup costs $O(1)$ work w.h.p. Ordered sets can be
stored in search trees called P-trees~\cite{blelloch2016just, sun2019join}.
Finding an element by rank or splitting a P-tree of $n$ elements takes $O(\log n)$ work~\cite{sun2018pam}. 
Inserting or deleting $k$ elements takes $O(k \log (1 + n/k))$ work and $O(\log n \log k)$ span~\cite{sun2019join}. A \emph{semisort}, taking an array of $n$ keyed elements and reordering
them so that elements with equal keys are contiguous,
can be computed in $O(n)$ expected work and $O(\log n)$ span w.h.p.~\cite{gu2015top}.
A \emph{minimum spanning forest (MSF)} is a spanning forest of minimum weight.
An MSF on $n$ vertices and $m$ edges can be statically
computed in $O(m)$ expected work and $O(\log n)$ span
w.h.p.~\cite{pettie2002randomized}.

\section{Parallel decremental MSF}\label{sec:dec-msf}

This section will show how to perform parallel batch-decremental MSF (supporting batches of edge deletions), and \cref{sec:dyn-msf} will show how to perform parallel batch-dynamic MSF.
We accomplish this by parallelizing the sequential dynamic
MSF algorithm by Holm, De Lichtenberg, and Thorup (HDT) that runs in $O(\log^4
n)$ amortized work per update~\cite{holm2001poly}. Their MSF algorithm has three
steps: first, they give an algorithm for dynamic connectivity; second, they
modify that algorithm into an algorithm for decremental MSF (parallelized in this section); and third, they use
decremental MSF to create a fully dynamic MSF algorithm (parallelized in \cref{sec:dyn-msf}). 
Without loss of generality, when discussing MSF, we assume edge
weights are unique by breaking ties using lexicographic ordering over edges'
endpoints.

The relevance of dynamic MSF to dynamic graph HAC is that single-linkage graph HAC 
can be solved with Kruskal's algorithm for computing a MSF after
negating all edge weights~\cite{gower1969minimum}.
A complication is that although the canonical output for HAC is a dendrogram, 
explicitly representing the dendrogram is too expensive for dynamic HAC since an edge update
can drastically change the dendrogram's structure (see \aref{app:dendro-single} for examples).
Instead, we implicitly represent the dendrogram by dynamically maintaining the MSF 
for the clustering.
We can then extract information about the single-linkage clustering from the MSF.
For instance, suppose that we want to answer the following ``group-by-cluster'' query,
a generalization of the type of query discussed in \cref{sec:intro}:
given a query set of $k$ vertices  $K \subseteq V$, we want to partition $K$ by the cluster 
that each vertex would be in if agglomerative clustering were run until all
similarities fell below a query similarity threshold $\theta$. We can answer such queries
in $O(k\log (1 + n/k))$ expected work and $O(\log n)$ span w.h.p.\ by
storing the MSF in a rake-compress (RC) tree,
computing a compressed path tree $P$ on the MSF relative to $K$
(\cref{sec:dyn-msf-bg} describes RC trees and compressed path trees), 
removing all edges with similarities below $\theta$ from $P$,
and computing connected components on $P$. 

\subsection{Background}

We first discuss existing algorithms and data structures that our work builds upon.

\paragraph{Euler tour trees.}
Euler tour trees (ETTs) are a
data structure for dynamic forests supporting edge insertion, edge
deletion, and connectivity queries in $O(\log n)$ deterministic work~\cite{henzinger1995randomized,miltersen1994complexity}. Tseng et al.\ introduce a
parallel batch-dynamic ETT that internally represents each tree in
in the forest as a circular skip list containing the tree's 
vertices and edges~\cite{tseng2019batch}.
%
The ETT can be augmented by a combining function $f: D \times D \to
D$, with $D$ being an arbitrary domain. After assigning values from $D$ to vertices
and edges, we can maintain the sum of $f$ over each tree (i.e., each connected component) in the forest 
by having each skip list node store the sum of $f$ over a contiguous subsequence 
of the sequence represented by the node's skip list.
Given an augmentation
function $f$ that takes $O(W)$ work and $O(S)$ span to compute, a batch of
$k$ insertions, $k$ deletions, or $k$ updates to assigned values for the augmentation
takes $O(Wk\log(1 + n/k))$ expected work and $O(S\log n)$ span w.h.p.\ on an
$n$-vertex forest.
The randomness in the bounds holds against oblivious adversaries
who cannot see heights of skip list elements.

\paragraph{Sequential dynamic connectivity.}
\label{sec:conn}
The HDT connectivity algorithm maintains a graph $G$ of $n$ vertices and
supports edge insertion, edge deletion, and connectivity queries.
The algorithm maintains $\log n$ \emph{levels}. Each edge is assigned a level,
and the algorithm maintains subgraphs $G_1 \subseteq G_2 \subseteq \ldots
\subseteq G_{\log n} = G$, where $G_i$ contains all edges of level at least $i$
and has the invariant that each connected component has size at most
$2^i$. The algorithm also maintains spanning forests $F_1 \subseteq F_2 \subseteq
\ldots \subseteq F_{\log n}$, where $F_i$ is a spanning forest of $G_i$. 
Connectivity queries are answered in $O(\log n)$ work by storing $F_{\log n}$ in
an ETT and querying the ETT. 
An edge insertion is handled by assigning the edge to level $\log n$ and inserting
it into $F_{\log n}$ if the edge's endpoints are not connected by a path.

A deletion of an edge $e$ of level $\ell$ is handled by deleting it from $G_i$ for
all $i \ge \ell$. If $e$ is not in $F_{\log n}$, then we are done. Otherwise,
the deletion of $e$ splits a connected component in $F_{\log n}$ in two, and we must
search for a \emph{replacement edge} reconnecting the two components. We
delete $e$ from $F_i$ for $i \ge \ell$ and conduct our search starting on level $\ell$. We
look at the smaller of the two connected components formed in $F_i$ by the
deletion of $e$. This connected component has size at most $2^{i-1}$, and so pushing
this entire component to level $i-1$ would not violate the size invariant. We
push all level-$i$ tree edges in the component to level $i-1$. Then,
we look at non-tree edges incident to the component one-by-one. If the non-tree
edge reconnects the two components, then we have found a replacement
edge---we change that edge into a tree edge, and we are done. Otherwise, we amortize
the cost of looking at this non-replacement edge by pushing it to level $i-1$.
We repeatedly run this search on increasing values of $i$ until a replacement edge is found.

If each $F_i$ is stored in an appropriately augmented ETT, then insertions and
deletions cost $O(\log^2 n)$ amortized work since each inserted edge can be
pushed down at most $\log n$ levels and it costs $O(\log n)$ work to find and push an
edge one level using the ETTs.


\paragraph{Parallel batch-dynamic connectivity.}
Acar et al.\ developed a parallel batch-dynamic version of the HDT algorithm~\cite{acar2019parallel}. We describe the ``non-interleaved''
version of their algorithm because we will modify it into a decremental MSF algorithm in \cref{sec:dec-msf-desc}. 
(The interleaved version has a better span bound, but it seems harder to adapt for decremental MSF.) 


The main difference from the original HDT algorithm to discuss is 
how the batch-parallel algorithm finds replacement edges after deleting a batch of edges.
The replacement search begins on
the minimum level among the deleted edges.
When searching on a level $i$, the algorithm proceeds in $O(\log n)$ rounds.
For every component of size at most
most $2^{i-1}$, we search for a replacement edge out of that component. To achieve
low span, each component performs a doubling search, looking at $2^j$
incident level-$i$ non-tree edges in parallel for increasing $j$ until finding a
replacement edge. We push non-replacement edges to the next level to amortize the cost
of examining them. We then compute a spanning tree over the replacement edges,
keeping only the replacement edges that are in the spanning tree. We
proceed to the next round on each ``active'' component, i.e., each component
that still has incident edges to search and that still has size at most
$2^{i-1}$. After all the rounds, we repeat at higher levels.

By storing each spanning forest in an ETT with appropriate augmentations, 
the algorithm can process a batch of $k$ edge updates in $O(k \log^2 n)$ expected amortized work and 
$O(\log^4 n)$ span w.h.p.
We note that the span bound can be tightened to $O(\log^3 n \log k)$. The bound has a $\log n$ term from the $O(\log n)$ rounds  per level, but since $k$ deletions creates $O(k)$ active
components and the active component count decreases geometrically each round, 
there are only $O(\log k)$ rounds per level.

\paragraph{Sequential decremental MSF} 
The HDT decremental MSF algorithm is initialized with a graph $G$ of $n$
vertices and maintains the MSF of $G$ while supporting edge deletions. There are
only two changes to the algorithm compared to the sequential HDT connectivity algorithm: we
initialize $F_{\log n}$ to be the MSF over $G$, and when we perform a replacement
search out of a component, we find the \emph{lightest} replacement edge by looking at
incident non-tree edges in increasing weight. Deletions still cost $O(\log^2 n)$
amortized work.

For correctness, the lightest replacement edge for a deleted edge must have the
minimum level among all possible replacement edges.  Holm et al.\ prove that the
algorithm maintains a \emph{cycle invariant} implying correctness: in every
cycle of $G$, the maximum-weight edge in the cycle is a non-tree edge and has
maximum level in the cycle. This invariant holds so long as whenever we push an
incident level-$i$ non-tree edge $e$ of a component to level $i-1$, $e$ is
lighter than the lightest replacement edge out of the component and we have already pushed
all lighter level-$i$ edges incident to this component.


\paragraph{Relative quantile summaries}
\label{sec:quantile}
Consider a set $S$ that is a subset of a totally ordered universe $U$. For an element $y \in U$, define
the rank of $y$ to be the number of elements in $S$ no greater than $y$:
$\rank(y; S) = \abs{\set{x \in S \mid x \le y}}$.
\noconf{We omit the second argument
of $\rank(\cdot;\cdot)$ when it is clear from context.}

For $\eps \in (0, 1)$ and a set $S$, an $\eps$-approximate relative
quantile summary $Q$ is a compressed form of the set that can compute queries
of the following form: given a rank $r \in [1, \abs{S}]$ such that $[r(1-\eps),
r(1+\eps)]$ contains an integer, return an element $y$ such that $\rank(y;S) \in
[r(1-\eps), r(1+\eps)]$. For the remainder of this paper, we use the
deterministic, mergeable relative quantile summaries
described by Zhang and Wang~\cite{zhang2007efficient}. \Aref{app:quantile-proofs-seq} re-derives the
construction of the summaries since Zhang and Wang's paper omits several proofs of correctness.

Additionally, we show in \aref{app:quantile-proofs-par} that we can parallelize operations on the
quantile summaries, which may be of independent interest. The following lemmas give the relevant operations and bounds.

\begin{lem}\label{lem:quantile-construct}
  Given a set $S$, if we can look up elements of $S$ by rank in $O(W)$ work, we
  can construct an $\ep$-approximate summary $Q$ of size 
  $\abs{Q} = O(\log(\ep \abs{S})/\ep)$
  in $O(W\log(\ep \abs{S})/\ep)$ work and $O(W)$ span.
\end{lem}
\begin{lem}\label{lem:quantile-query}
  Given an approximate summary $Q$ on set $S$, we can answer queries in $O(\log
  \abs{Q})$ work and can obtain the minimum element of $S$ in constant work.
\end{lem}
\begin{lem}\label{lem:quantile-combine}
  Given an integer $b > 0$ and two $\ep$-approximate summaries $Q_1$ and $Q_2$ on
  non-overlapping sets $S_1$ and $S_2$, we can
  create an $(\ep + 1/b)$-approximate summary $Q$ over $S_1 \cup S_2$ of size
  $O(b \log (n/b))$ in $O(\abs{Q_1} +
  \abs{Q_2})$ work and $O(\log(\abs{Q_1} + \abs{Q_2}))$ span.
\end{lem}
We let $\combine(Q_1, Q_2, b)$ denote the algorithm combining summaries $Q_1$ and
$Q_2$ with parameter $b$.

\subsection{Finding light replacement edges}\label{sec:ett}

Like how the HDT decremental MSF algorithm comes from modifying the HDT dynamic connectivity algorithm,
we will obtain a parallel batch-decremental MSF algorithm by modifying Acar et al.'s batch-dynamic connectivity
algorithm to search for batches of light replacement edge candidates rather than arbitrary candidates.
The primary challenge is searching for the lightest non-tree edges incident on a component efficiently in parallel.

As in Acar et al's connectivity algorithm,
for each HDT level $\ell \in [1, \log n]$, we store 
$F_\ell$ in an augmented batch-dynamic ETT.
Using quantile summaries, we add an additional augmentation, described
in the proof of the following theorem, that allows fetching the lightest
non-tree edges incident on a component at the cost of increasing the
running time of edge insertions and deletions for the ETTs. \Cref{sec:ett-aug-example} walks through an example of the additional augmentation.

\begin{theorem}\label{thm:ett-aug}
  Let $\ell \in [1, \log n]$, and for each vertex $v \in V$, let $N_{v,\ell}$
  represent the level-$\ell$ non-tree edges incident to $v$. Using $O(n \log^2n)$
  space w.h.p., we can support the following operations over $F_\ell$, with all
  work bounds being in expectation and all span bounds being w.h.p.:
  \begin{itemize}
    \item inserting or deleting $k$ edges to $F_\ell$ in $O(k \log^2n \log(1 + n/k))$
      work and $O(\log n \log \log n)$ span,
    \item inserting or deleting $k$ edges to $\set{N_{v, \ell}}_{v \in V}$ in $O(k \log^2n \log(1 + n/k))$
      work and $O(\log n (\log \log n + \log k))$ span,
    \item obtaining the $k(1 \pm 1/2)$ lightest edges in
      $\bigcup_{v \in C} N_{v, \ell}$ of a connected component $C$ of $F_\ell$
      in $O(k \log n)$ work and $O(\log n)$ span.
  \end{itemize}
\end{theorem}
\begin{proof}
  We store $F_\ell$ in batch-dynamic ETTs. We will augment each ETT skip list
  node with a quantile summary $Q$ and an integer $t \ge 0$, where
  $Q$ summarizes the weights of non-tree edges 
  incident on the vertices in the node's subsequence and 
  $t$ indicates the error $Q$ has accumulated
  from $\combine(\cdot, \cdot, \cdot)$ operations.

  The augmented value for a vertex $v$ is given by constructing a $1/4$-approximate 
  quantile summary $Q$ over $N_{v,\ell}$
  and setting $t = 0$. Each vertex $v$
  stores $N_{v, \ell}$ in an ordered set. A weight-$w$ edge in $N_{v, \ell}$ to a neighbor
  $u$ is stored as a tuple $(w, v, u)$, and ordering is lexicographic.
  By \cref{lem:quantile-construct}, using a P-tree to represent the ordered set, we can construct $Q$
  in $O(\log^2 n)$ work and $O(\log n)$ span.

  Define $b(t) = 8(\log n + t^2/\log n)$ and
  define the ETT augmentation function $f$ as
  \begin{align*}
    f((Q_1, t_1), (Q_2, t_2)) &=
    \left(
      \combine(Q_1, Q_2, b(t)), t
    \right) \\
    \text{where } t &= \max\set{t_1, t_2} + 1
  .\end{align*}

  For a skip list node in the ETT whose subsequence has vertices $S$, its
  augmented value $Q$ summarizes the weights of $\bigcup_{v \in S}
  N_{v,\ell}$, and its $t$ is bounded by the longest search path length from that
  node to a node representing some $v \in S$ at the bottom level of the skip list. Since the maximum path length in a skip list is $O(\log n)$
  w.h.p.~\cite{demaine2004lecture}, we have $t = O(\log n)$ and $b(t) = O(\log n)$ for every summary w.h.p.
  By \cref{lem:quantile-combine}, the augmentation takes $O(\log^2 n)$ space per skip list node and runs in
  $O(\log^2 n)$ work and $O(\log \log n)$ span, all w.h.p. 
  
  Recall that given an augmentation function that costs $O(W)$ work and $O(S)$ span,
  a batch of $k$ updates to an ETT takes
  $O(Wk\log(1 + n/k))$ expected work and $O(S\log n)$ span w.h.p.
  Therefore, with our augmentation function, a batch of $k$ edge insertions or deletions to $F_\ell$ takes 
  $O(k \log^2n \log(1 + n/k))$ expected work and $O(\log n \log \log n)$ span w.h.p. 
  The cost of insertions or deletions to $\set{N_{v,\ell}}_{v \in V}$ also incurs the same cost
  in updating augmented values, but there is the additional cost of 
  having to actually update $\set{N_{v,\ell}}_{v \in V}$ and to rebuild the quantile summaries
  over $\set{N_{v,\ell}}_{v \in V}$.
  
  To update $\set{N_{v,\ell}}_{v \in V}$ with $k$ edges $U$,
  apply a semisort to group the edges by endpoint:
  let $r$ be the number of distinct endpoints in $U$, and
  let $K = \set{(v_1,
  E_1), \ldots, (v_k, E_r)}$ represent the semisorted updates, where
  $N_{v_i,\ell}$ should be updated with edges $E_i$ for each integer $i \in [1, r]$.
  Updating the ordered set for $N_{v_i,\ell}$ costs
  $O(\abs{E_i} \log(1 + \abs{N_{v_i,\ell}}/\abs{E_i}))$ work and $O(\log n \log k)$ span. 
  The sum of this work over all $v_i$ is $O(k \log
  n)$. We then rebuild the quantile summary for each $v_i$ from scratch via \cref{lem:quantile-construct}
  in $O(r \log^2 n)$ total work and $O(\log n)$ span. Adding the cost of updating augmented values gives a
  total expected work is $O(k \log^2n \log(1 + n/k))$ and 
  the total span is $O(\log n (\log \log n + \log k))$ w.h.p.

  The approximation error of a summary $Q$ in the ETT is bounded above by 
  \begin{align*}
    \frac{1}{4} + \sum_{t=1}^\infty \frac{1}{b(t)}
    =
    \frac{1}{4} +
      \sum_{t=1}^{\log n} \frac{1}{b(t)}
      + \sum_{t=\log n + 1}^{\infty} \frac{1}{b(t)}
    \\ <
    \frac{1}{4} +
      \sum_{t=1}^{\log n} \frac{1}{8\log n}
      + \sum_{t=\log n + 1}^{\infty} \frac{1}{8t^2/\log n}
    \\=
    \frac{1}{4} + \frac{1}{8}
      + \frac{\log n}{8} \sum_{t=\log n + 1}^{\infty} \frac{1}{t^2}
    <
    \frac{3}{8} +
      \frac{\log n}{8} \int_{t=\log n}^{\infty}\frac{1}{t^2}\,dt
     = \frac{1}{2}
  .\end{align*}
  That is, $Q$ is always a $1/2$-approximate quantile summary. We can therefore fetch
  the $k(1 \pm 1/2)$ lightest edges of $\bigcup_{v \in C} N_{v, \ell}$ for a
  connected component $C$ by querying the summary of component $C$ for a
  weight $w$ whose rank is $k$. Then, by checking whether the summaries'
  minimum element is less than $w$, we traverse down the skip list to efficiently
  find all vertices $v$ in $C$ such that $N_{v,\ell}$ has edges
  lighter than $w$ in $O(k \log(1 + n/k))$ expected work and $O(\log n)$ span
  w.h.p. We fetch those edges from each vertex $v$ by splitting the
  ordered set for $N_{v, \ell}$ in $O(\log n)$ work
  and $O(\log n)$ span per vertex for $O(k \log n)$ total work.
\end{proof}

\subsubsection{Example of augmentation.}\label{sec:ett-aug-example}

\begin{figure*}
    \centering
    \begin{subfigure}{0.2\linewidth}
      \centering
      \includegraphics[scale=0.3]{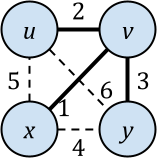}
      \caption{Graph and a spanning forest (in bold). 
      Each non-tree edge has the same HDT level $\ell$.}
    \end{subfigure}%
    \begin{subfigure}{0.8\linewidth}
        \centering
        \includegraphics[scale=0.3]{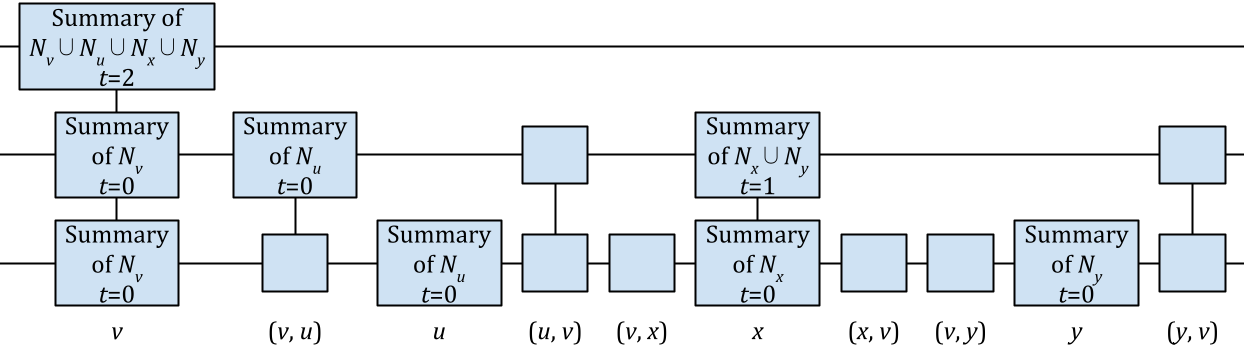}
        \caption{Skip list of the ETT for the graph.}
    \end{subfigure}
    \caption{The figure displays a graph and a possible height-$3$ skip list corresponding
      to an ETT representing the spanning forest of the graph. The skip list is circular, so the last
      node at each skip list level is connected to the first node at the same level.
      The labels below the skip list nodes at
      the bottom skip list level denote the vertex or edge that the node represents.
      At higher levels, the skip
      list nodes represent a contiguous subsequence of the overall sequence (e.g., the fourth
      node at the middle level represents $(x, (x,v), (v,y), y)$, and the node at the top level
      represents the whole sequence). 
      Following the augmentation of \cref{thm:ett-aug},
      each list node is augmented with a quantile summary
      for the vertices within its subsequence.
      In the figure, the labels inside each skip list node are the augmented values for each node. Some nodes contain no label because their subsequence contains no vertices.
    }
    \label{fig:ett-aug}
\end{figure*}

This subsection gives an example illustrating
the ETT augmentation from \cref{thm:ett-aug}.
\Cref{fig:ett-aug} displays an example graph $G = (V, E)$ with one connected component, a spanning tree for the graph, 
and a possible skip list internally held by an ETT representing the tree.
For simplicity, every non-tree edge in the graph has the same HDT level $\ell$, and we write $N_v$ instead of $N_{v,\ell}$ to denote the non-tree edges incident on a vertex $v$.

The skip list contains a sequence given by taking an Euler tour on the spanning tree after duplicating each edge in both directions and after adding a self-loop edge on every vertex.
At the bottom level of the skip list, each node representing some vertex $v \in V$ (i.e., representing the self-loop edge for vertex $v$)
constructs a quantile summary over $v$'s incident non-tree edges $N_v$. 
For instance, $N_u = \set{(5,u,x), (6,u,y)}$, and the skip list node for $u$ is augmented 
with a $1/4$-approximate quantile summary on $N_u$.

The quantile summaries at higher levels of the skip list are computed by 
calling the augmentation function $f$ to combine quantile summaries at the level below, ignoring
skip list nodes that correspond to edges and hence have no quantile summaries. For instance, the 
node at the top level of the skip list computes its quantile summary over $N_v \cup N_u \cup N_x \cup N_y$ 
by calling $f$ on the 
quantile summaries for $N_v$ and $N_u$ and then calling $f$ on the result of the previous call along with 
the quantile summary for $N_x \cup N_y$. Each quantile summary has an associated $t$ value
such that the summary is $(1/4 + \sum_{i=1}^t 1/b(i))$-approximate, where $b(\cdot)$ is defined in
the proof of \cref{thm:ett-aug}. The $\sum_{i=1}^t 1/b(i)$ term is accumulated from $\combine$ calls
used to compute $f$.
When taking a union of incident non-tree edges such as $N_x \cup N_y$, we may
have an edge twice, once in each direction, e.g., $N_x \cup N_y = \set{ 
  (4, x, y), (4, y, x), (5, x, u), (6, y, u)
}$ has the weight-$4$ edge $\set{x, y}$ twice. 
For our use case of decremental MSF (\cref{sec:dec-msf-desc}), 
the edge duplication does not affect correctness or running time complexity.

If there were multiple nodes at the top level of the skip list, we would also call $f$
to combine all of the quantile summaries at the top level to obtain a quantile summary for the entire
connected component. At an arbitrary top-level node, we would store a pointer to that quantile summary so that we can quickly
fetch a quantile summary for this connected component.

As an example of updating non-tree edges, suppose we wanted to delete an element from $N_y$. 
We perform the deletion on $N_y$, rebuild the quantile summary for $N_y$ entirely from
scratch, and then apply $f$ again to rebuild the quantile summaries for the fourth node at the middle
skip list level and for the node at the top level.

As an example of searching for light non-tree edges, suppose that we wanted to find the lightest $k(1 \pm 1/2)$
non-tree edges incident on this component for some $k$. Define the \emph{children} of a skip list node 
with associated subsequence $S$ to be the nodes at the level immediately below whose 
subsequences constitute $S$, e.g., the children of the fourth node at the 
middle level are the sixth through ninth (inclusive) nodes at the bottom level. 
We first query the quantile summary at the top level of the skip list for a rank-$k$ entry.
Suppose that it returned $(4,y,x)$. Then, we traverse down to the middle level to inspect the quantile
summaries of the node's children. At every quantile summary whose
minimum element is at most $(4,y,x)$ lexicographically, we traverse down to \emph{that} node's children.
In this example, only the fourth node at the middle level satisfies this condition.
We again check childrens' quantile summaries' minimum elements, and in this case, the
sixth and ninth nodes at the bottom level satisfy the condition. Since we have
reached the bottom of the list, we directly access $N_x$ and $N_y$ and fetch all elements
that are at most $(4, y, x)$.

\subsection{Parallel batch-decremental MSF}\label{sec:dec-msf-desc}

As with the sequential HDT decremental MSF algorithm,
two changes are needed to change Acar et al.'s batch-dynamic connectivity
algorithm into a batch-decremental MSF algorithm.
First, given an input graph $G$, we compute an MSF $F$
over $G$ and set $F_{\log n} = F$. The MSF for the graph will always
be $F_{\log n}$.

Second, when performing a doubling search out of a component to find a replacement edge,
instead of looking for $2^j$ arbitrary incident
non-tree edges on phase $j$ of a doubling search,
we use the ETT augmentation from \cref{thm:ett-aug}
to search for the $2^j(1 \pm 1/2)$ lightest incident non-tree edges. 
To maintain the HDT cycle invariant, we do not push any
edges on a doubling phase in which we find a replacement edge. In addition,
to reduce span, we defer pushing edges to the end of 
the entire replacement search on a level
rather
than pushing non-tree edges after every doubling phase.

\begin{theorem}\label{thm:dec-msf}
  We can initialize a
  batch-decremental MSF data structure in
  $O(\log^2 n)$ span w.h.p. The data structure supports
  batches of $k$ edge deletions in
  $O(\log^3 n \log k)$ span w.h.p.\ and uses $O(m + n\log^3 n)$
  space w.h.p. The total expected work across initialization and all deletions
  is $O(m \log^3 n \log (1 + n/\Delta))$, i.e., $O(\log^3 n \log (1 + n/\Delta)) \le O(\log^4 n)$ amortized per edge,
  where $\Delta$ is the average batch size across all batch deletions.
\end{theorem}
\begin{proof}
  For initialization, computing $F_{\log n}$ costs
  $O(m)$ expected work and $O(\log n)$ span w.h.p.
  Then, by \cref{thm:ett-aug}, 
  storing $F_{\log n}$ in an
  augmented ETT and updating the ETT with $O(m)$ incident non-tree edges
  costs $O(m\log^2 n)$ expected work and $O(\log^2 n)$ span w.h.p.

  The span to delete a batch of $k$ edges remains the same as Acar et al.'s
  algorithm. Despite the increase in span that our more complicated ETT
  augmentation incurs for insertions and pushing edges, the span is still
  dominated by the doubling search, whose span remains the same. 
  
  On the other hand, the work increases by a factor of $O(\log^2n)$ to a total of $O(\log^4
  n)$ expected amortized work per edge due to the increased work for ETT
  insertion and pushing. 
  The $(1 \pm 1/2)$ uncertainty in searching for incident
  non-tree edges may increase the amount of amortized cost to charge to each
  edge by a constant factor, but this does not affect the asymptotic bounds.
  Not pushing edges found on the last phase of a doubling search also
  only affects amortized costs by a constant factor. 
  Analysis in \aref{app:msf-work-bounds} gives the tighter 
  $O(\log^3 n \log (1 + n/\Delta)$ work bound.

  The maximum space usage is $O(n \log^2 n)$ w.h.p.\ for each of the $\log n$
  ETTs plus $O(m)$ total space to store the non-tree edges.
\end{proof}

\section{Parallel fully dynamic MSF}\label{sec:dyn-msf}

In this section, we describe
a parallel batch-dynamic MSF algorithm supporting both batch insertions and batch deletions of edges (\cref{sec:dyn-msf-desc}) and provide an example execution of the algorithm (\cref{sec:dyn-msf-example}).

\subsection{Background}\label{sec:dyn-msf-bg}

\paragraph{Compressed path trees}
Like ETTs, \emph{rake-compress (RC) trees} are a parallelizable data structure for dynamic
forests~\cite{acar2004dynamizing,acar2020parallel}. 
Inserting or deleting $k$ edges from an RC tree takes $O(k
\log (1 + n/k))$ expected work and $O(\log^2 n)$ span w.h.p.

Anderson et al. showed that if a dynamic forest $F$ of $n$ vertices is stored in
an RC tree, 
then given $k$ vertices in $F$ (``marked'' vertices),
we can construct a compressed form of $F$ called
a \emph{compressed path tree} relative to the vertices in $O(k\log(1+n/k))$ expected work
and $O(\log n)$ span w.h.p.~\cite{anderson2020work}. The compressed
path tree is a forest $F'$ on $O(k)$ vertices (including all marked vertices) such that the heaviest edge on the path between
any pair of marked vertices has the same weight in $F'$ as in $F$.
More specifically, the compressed path tree is the union of the paths between the
marked vertices with all non-marked vertices of degree below three spliced out.
Given $k$ edges in $F$, we can also find
which compressed edges in $F'$ correspond to those edges in $O(k
\log (1 + n/k))$ expected work and $O(\log n)$ span w.h.p.\ by traversing up the
RC tree that generated $F'$.

\paragraph{Dynamic MSF}
The fully dynamic HDT MSF algorithm supports both insertions and
deletions with $O(\log^4 n)$ amortized work per update. In describing this
algorithm, we follow the presentation of Holm, Rotenberg, and
Wulff--Nilsen~\cite{holm2015faster}.

Along with maintaining the MSF $F$ of the graph $G$ (the \emph{global} tree
and graph), the algorithm maintains $2\log n + 1$ subgraphs $A_0,\ab A_1,\ab
\ldots,\ab A_{2 \log n} \subseteq G$ and MSFs $F_0,\ldots,F_{2 \log n}$ of each subgraph (the
\emph{local} graphs and trees). Each
$A_i$ has at most $2^i$ non-tree edges $A_i \setminus F_i$ (the \emph{edge-count invariant}), and
each non-tree edge in $G$ is a non-tree edge of some $A_i$ (the \emph{non-tree-edge invariant}). We maintain
decremental MSF data structures over each local graph.

To insert an edge $e = \set{u, v}$, we query for the heaviest
edge $e'$ on the path between $u$ and $v$ in $F$ by storing $F$ in a top
tree~\cite{alstrup1997minimizing}. We replace $e'$
with $e$ in $F$ if $e$ is lighter. Either $e$ or $e'$ now becomes a new non-tree
edge. To make the non-tree-edge invariant hold for the edge, we call
\tsc{Update}, a subroutine that we describe shortly below, on the edge to insert it
as a local non-tree edge.

To delete an edge $e$, we delete $e$ from all local graphs
and obtain a set of $O(\log n)$ local replacement edges $R$. 
If $e$ is in $F$, we delete it from $F$ and need a global replacement edge.
Due to the non-tree-edge invariant,
the lightest edge $r$ in $R$ reconnecting $F$ is the global replacement edge.
We insert that edge into the global tree. Since edges in $R$ (besides $r$) are global
non-tree edges that might now violate the non-tree edge invariant, we call
\tsc{Update} on $R$.

The \tsc{Update} subroutine with input $U$ inserts the edges
in $U$ as local non-tree edges. It re-initializes $A_j$ to be $F \cup U \cup
\bigcup_{i \le j}(A_i \setminus F_i)$, with $j$ being the minimal value such that this
reinitialization respects the edge-count invariant. The new local tree edges for $A_j$ are the edges in $F$,
and the other edges become local non-tree edges. The subroutine then clears $A_i$ for all $i < j$.

The number of tree edges in each $A_i$ may be large, and so we only store them
in compressed form. When initializing $A_i$, we use a top tree to efficiently compute a structure similar
to a compressed path tree. Initializing and storing $A_i$ then takes only $O(2^i
\log n)$ work and $O(2^i)$ space, and initializing a decremental MSF over $A_i$
costs $O(2^i \log^2 2^i)$ amortized work.

To analyze the work, in \tsc{Update}, the choice of $j$ means that there
are at least $2^{j-1}$ non-tree local edges $U \cup \bigcup_{i < j}(A_i
\setminus F_i)$ being pushed up to $A_j$. These edges pay for the initialization cost of $A_j$. A
non-tree edge costs $O(\log^3 n)$ across its lifetime since it can be pushed up $2 \log n$ times and may pay $O(\log^2 n)$
amortized work on each push to pay for the cost per edge in the newly initialized decremental MSF
data structure. Since each global deletion introduces $O(\log n)$
non-tree local edges, 
the amortized cost of a deletion in the dynamic MSF algorithm is $O(\log^4 n)$.


\subsection{Parallel batch-dynamic MSF}\label{sec:dyn-msf-desc}
Our parallel batch-dynamic MSF algorithm comes from
parallelizing the fully dynamic HDT MSF
algorithm. The main changes are to use our decremental MSF algorithm from
\cref{sec:dec-msf-desc} and to use RC trees instead of top trees
for compressing local graphs and for efficient batch insertion.

\begin{algorithm}
  \captionsetup{font=footnotesize}
  \caption{The algorithm that sets global variables to initialize the batch-dynamic MSF data structure on an
  $n$-vertex graph.}
  \label{alg:msf-init}
\begin{algorithmic}[1]
          \figfont
  \Procedure{Initialize}{$n$}
    \State $F \gets$ RC tree on an empty $n$-vertex graph
      \label{ln:msf-init:f}
      \Comment The MSF, i.e., the global tree.
    \Pfor{$i = 0, 1, 2, \ldots, 2\log n$}
      \State $A_i \gets \varnothing$ \Comment Decremental MSF data structure for the $i$-th local graph.
        \label{ln:msf-init:a}
      \State $T_i \gets$ RC tree on an empty $n$-vertex graph 
        \label{ln:msf-init:t}
      \State $(B_{D, i}, B_{I, i}) \gets (\varnothing, \varnothing)$
        \label{ln:msf-init:b}
    \EndFor
  \EndProcedure
\end{algorithmic}
\end{algorithm}

\Cref{alg:msf-init} initializes the data structure on an $n$-vertex graph. All variables
(and only these variables) defined in this algorithm are globally visible.
We assume the input graph begins with no edges 
since input edges can be added separately via batch insertion. The RC tree $F$ maintains the MSF (global tree) (\cref{ln:msf-init:f}). 

Each $A_i$ is a batch-decremental MSF data structure over the 
$i$-th local graph, which is initially empty (\cref{ln:msf-init:a}). Whenever we initialize $A_i$, we will 
need to compress its tree edges by computing a compressed path tree on the tree edges relative to 
$A_i$'s non-tree edges' endpoints. The RC tree used to compute the compressed path tree should remain
unmodified until $A_i$'s next initialization so that when deleting edges, we can 
use the RC tree to look up the compressed representations of the edges in $A_i$. The RC tree
$T_i$ serves this purpose for $A_i$ (\cref{ln:msf-init:t}). 
Its value matches $A_i$'s (uncompressed) tree edges at $A_i$'s latest initialization, or equivalently,
the value of $F$ at $A_i$'s latest initialization.
To update $T_i$ to match $F$ at $A_i$'s next initialization, we keep buffers $B_{D, i}$ and $B_{I, i}$ representing the difference between  
$T_i$ and $F$ (\cref{ln:msf-init:b}). In particular, $(T_i \setminus  B_{D, i}) \cup B_{I, i} = F$.
\Cref{sec:dyn-msf-example} illustrates an example of how $A_i$ and $T_i$ changes over several edge 
updates.

\begin{algorithm}
  \captionsetup{font=footnotesize}
  \caption{A helper algorithm for restoring the HDT non-tree-edge invariant.}
  \label{alg:msf-update}
\begin{algorithmic}[1]
  \figfont
  \Procedure{Update}{$U = \set{(\set{u_1, v_1}, w_1), \ldots, (\set{u_k, v_k}, w_k)}$}
    \For{$i = 0, 1, 2, \ldots, 2\log n$} \label{ln:msf-update:for}
      \State $U \gets U \cup (\Call{NontreeEdges}{A_i} \setminus F)$
        \label{ln:msf-update:u}
      \State $A_i \gets \varnothing$
        \label{ln:msf-update:clear-a}
      \If{$\abs{U} \le 2^i$} \label{ln:msf-update:if}
        \State $T_i.\Call{Delete}{B_{D,i}}$
          \label{ln:msf-update:tdel}
        \State $T_i.\Call{Insert}{B_{I,i}}$
          \label{ln:msf-update:tins}
        \State $(B_{D, i}, B_{I, i}) \gets (\varnothing, \varnothing)$
          \label{ln:msf-update:clear-b}
        \LComment \parbox[t]{.8\linewidth}{
          Given an RC tree $T$ representing a forest, $T.\Call{CompressedPathTree}{\cdot}$
          takes a list of vertices $L$ and returns a compressed path tree for the forest relative to $L$.
        }
        \Local $P \gets T_i.\Call{CompressedPathTree}{\bigcup_{(\set{u, v}, w) \in U}\set{u, v}}$
          \label{ln:msf-update:p}
        \State $A_i \gets$ Batch-decremental MSF on $P \cup U$
          \label{ln:msf-update:a}
        \Break
      \EndIf
    \EndFor
      \label{ln:msf-update:for-end}
  \EndProcedure
\end{algorithmic}
\end{algorithm}

Before discussing batch insertion, we describe the helper function
\tsc{Update} (\cref{alg:msf-update}) that takes non-tree edges $U$ and 
inserts them in a local graph to satisfy the HDT non-tree-edge invariant. We iterate
through each local graph $A_i$ sequentially to find some $A_j$ to re-initialize
with $U$ such that the HDT edge-count invariant still holds (\cref{ln:msf-update:for}).
As we iterate through increasing $i$, we collapse the non-tree edges of $A_i$
into $U$ (\crefrange{ln:msf-update:u}{ln:msf-update:clear-a}) since pushing them up to level $j$ will pay
for the re-initialization cost of $A_j$. We discard the tree edges of
$A_i$ since they are irrelevant to the non-tree-edge invariant.
Once we find the level $j$ (\cref{ln:msf-update:if}), we update the RC tree $T_j$ to match
the global tree $F$ using buffers $B_{D,j}$ and $B_{I, j}$ and clear the
buffers (\crefrange{ln:msf-update:tdel}{ln:msf-update:clear-b}). Then, we use
$T_j$ to compress $F$ into a compressed path tree $P$ relative to $U$ so that
$O(\abs{P}) = O(2^j)$, and we set $A_j$ to be a newly initialized decremental 
MSF data structure over $P \cup U$ (\crefrange{ln:msf-update:p}{ln:msf-update:a}).
(The compressed local graphs may be non-simple because between a pair of vertices, there can be both
one compressed tree edge and one non-tree edge. The decremental MSF algorithm still works in
this setting.)

\begin{algorithm}
  \captionsetup{font=footnotesize}
  \caption{The algorithm for inserting a batch of edges.}
  \label{alg:msf-insert}
\begin{algorithmic}[1]
  \figfont
  \Procedure{BatchInsert}{$U = \set{(\set{u_1, v_1}, w_1), \ldots, (\set{u_k, v_k}, w_k)}$}
    \Local $P \gets F.\Call{CompressedPathTree}{\bigcup_{i=1}^k \set{u_i, v_i}}$
      \label{ln:msf-insert:p}
    \Local $M \gets \Call{MSF}{P \cup U}$
      \label{ln:msf-insert:m}
    \Local $(D, I) \gets (P \setminus M, U \cap M)$
      \label{ln:msf-insert:di}
    \State $F.\Call{Delete}{D}$
      \label{ln:msf-insert:del}
    \State $F.\Call{Insert}{I}$
      \label{ln:msf-insert:ins}
    \Pfor{$i = 0, 1, 2, \ldots, 2\log n$}
     \State $(B_{D, i}, B_{I, i}) \gets (B_{D, i} \cup (D \setminus B_{I, i}), (B_{I, i} \setminus D) \cup I)$
      \label{ln:msf-insert:set-b}
    \EndFor
    \State $\Call{Update}{D \cup (U \setminus M)}$
      \label{ln:msf-insert:up}
  \EndProcedure
\end{algorithmic}
\end{algorithm}

\Cref{alg:msf-insert} gives pseudocode for batch insertion. We start by compressing the
global tree $F$ into a compressed path tree $P$ relative to inserted edges $U$
(\cref{ln:msf-insert:p}). 
Each compressed edge $e$ in the compressed path tree also stores a pointer to the
heaviest edge in the path that $e$ represents in $F$. In this way we can, for
brevity, refer to edges from the compressed path tree and the corresponding heavy
edges in the uncompressed tree $F$ interchangeably in the
pseudocode. We compute an MSF $M$ over $P \cup U$ (\cref{ln:msf-insert:m}). 
Using $M$, we can determine which edges $I$ from $U$ to insert into the global
tree $F$ and which edges $D$ from $F$ get replaced by $I$
(\cref{ln:msf-insert:di}). We delete $D$ from $F$, insert $I$ into $F$, and
update the buffers for every local graph $A_i$
(\crefrange{ln:msf-insert:del}{ln:msf-insert:set-b}). Finally, we call
\tsc{Update} on edges $D \cup (U \setminus M)$ since they are new global non-tree
edges that may violate the non-tree-edge invariant (\cref{ln:msf-insert:up}).

\begin{algorithm}
  \captionsetup{font=footnotesize}
  \caption{The algorithm for deleting a batch of edges.}
  \label{alg:msf-delete}
\begin{algorithmic}[1]
  \figfont
  \Procedure{BatchDelete}{$U = \set{\set{u_1, v_1}, \ldots, \set{u_k, v_k}}$}
    \Local $D \gets F \cap U$
      \label{ln:msf-delete:d}
    \State $F.\Call{Delete}{D}$
    \Pfor{$i = 0, 1, 2, \ldots, 2\log n$}
      \State $(B_{D, i}, B_{I, i}) \gets (B_{D, i} \cup (D \setminus B_{I, i}), B_{I, i} \setminus D)$
        \label{ln:msf-delete:set-b}
      \Local $U' \gets$ Representation of $U$ in compressed $A_i$
        \label{ln:msf-delete:rep}
      \LComment \parbox[t]{.85\linewidth}{$A_i.\Call{Delete}{\cdot}$ takes a list of edges, deletes them from $A_i$, and 
        returns the replacement edges used to reconnect $A_i$.}
      \Local $R_i \gets A_i.\Call{Delete}{U'}$
          \label{ln:msf-delete:ri}
    \EndFor
    \State \Call{BatchInsert}{$\bigcup_{i=0}^{2\log n} R_i$}
      \label{ln:msf-delete:ins}
  \EndProcedure
\end{algorithmic}
\end{algorithm}

\Cref{alg:msf-delete} describes batch deletion. First, we delete the input
edges $U$ from the global tree $F$ and update the buffers 
for every local graph $A_i$ accordingly (\crefrange{ln:msf-delete:d}{ln:msf-delete:set-b}).
In parallel over every $A_i$, we want to delete $U$ from $A_i$, though this requires some effort 
since $A_i$ is in compressed form. To map each edge $e$ in $U$ to its representation in 
$A_i$ (\cref{ln:msf-delete:rep}), there are three cases: 
$e$ appears in compressed form in $A_i$ (because it was a tree edge in $A_i$ when $A_i$ 
was last initialized),
$e$ appears in uncompressed form in $A_i$, or it does not exist in $A_i$. 
To handle the first case, we try looking up $e$ in $T_i$ to get a compressed 
edge and then try looking that up in $A_i$.
Simultaneously, we try looking up $e$ directly in $A_i$ to handle the second case. 
If the two cases fail, then we are in the third case and ignore the edge.

Now we can delete $U$ from $A_i$ and extract the local replacement edges $R_i$ that 
the decremental MSF data structure uses to replace $U$ (\cref{ln:msf-delete:ri}). 
Finally, we insert the replacement edges into the global tree by calling
\tsc{BatchInsert} (\cref{ln:msf-delete:ins}). Although these replacement edges
are already global edges, calling \tsc{BatchInsert} has the correct
behavior of reconnecting $F$ and calling \tsc{Update}.

\begin{theorem}
  Our dynamic MSF algorithm
  maintains an MSF in $O(k \log^6 n)$ expected amortized work for a batch of $k$
  edge insertions or $k$ edge deletions. Insertions take $O(\log^2 n)$ span w.h.p.,
  and deletions take $O(\log^3 n \log k)$ span w.h.p. The maximum amount of space
  the data structure uses is $O(m + \min\set{m, n\log n} \log^3 n)$ 
  w.h.p., where $m$ is the maximum number of edges in the graph.
\end{theorem}
\begin{proof}

  \emph{Work}: We first analyze the work for \tsc{Update}
  (\cref{alg:msf-update}). We will give $(2 \log n - i) \cdot
  O(\log^4 n)$ amortization credits to non-tree edges in local graph $A_i$ and
  get $O(k \log^5 n)$ expected amortized work for \tsc{Update} as a consequence.

  Suppose we call \tsc{Update} (\cref{alg:msf-update})
  with $k$ input edges. Let $\bar{U}$ be the original input to \tsc{Update}, and let $U_i$
  indicate the value of $U$ after the $i$-th iteration of
  \cref{ln:msf-update:u}. We give $O(\log^5 n)$ 
  credits to each edge in $\bar{U}$. Let $j$ be
  the value of $i$ that satisfies the condition on \cref{ln:msf-update:if}.
  The work done by insertions to $U$ (\cref{ln:msf-update:u}) across all
  iterations sums to $O(\smallabs{\bar{U}} + 2^j)$. Updating $T_j$ on
  \crefrange{ln:msf-update:tdel}{ln:msf-update:tins} costs $O((\smallabs{B_{D,
  j}} + \smallabs{B_{I, j}}) \log n)$ expected work, which we charge to the \tsc{BatchInsert}
  and \tsc{BatchDelete} calls that inserted these elements into $B_{D, j}$ and $B_{I, j}$.
  Computing the compressed path tree (\cref{ln:msf-update:p}) costs
  $O(\smallabs{U_j}\log n)$ expected work, and initializing the decremental MSF data
  structure (\cref{ln:msf-update:a}) costs $O(\smallabs{U_j}\log^4 n)$ expected
  amortized work.
  We know that $\smallabs{U_j} \le 2^j$ and $\smallabs{U_{j-1}} > 2^{j-1}$ due to the
  choice of $j$. We pay for the $O(2^j + \smallabs{U_j}\log^4 n) = O(2^j\log^4 n)$ work of $\tsc{Update}$
  by charging $O(\log^4 n)$ credits to the elements in $\smallabs{U_{j-1}}$ that we have
  pushed up to local graph $A_j$. The remaining expected amortized
  work is $O(k\log^5 n)$ from the credits we gave to $\bar{U}$.

  Batch insertion (\cref{alg:msf-insert}) also costs $O(k \log^5 n)$ amortized
  expected work due to its work being dominated by \tsc{Update}
  (\cref{ln:msf-insert:up}). For instance, RC tree operations and
  computing an MSF takes only $O(k \log (1 + n/k))$ expected work on
  \crefrange{ln:msf-insert:p}{ln:msf-insert:ins}.
  Updating buffers $B_{i,*}$ on \cref{ln:msf-insert:set-b} costs only $O(k
  \log^2 n)$ amortized total work, where one $\log n$ factor comes from summing
  over $i$ and the other $\log n$ factor pays for the cost of updating
  $T_i$ with $B_{i,*}$ in \tsc{Update}.

  Batch deletion (\cref{alg:msf-delete}) costs $O(k \log^6 n)$ amortized
  expected work. Like with batch insertions, the expected work of RC tree operations and
  dictionary operations is $O(k \log^2 n)$ on
  \crefrange{ln:msf-delete:d}{ln:msf-delete:rep}. We charge the cost of
  deleting $k$ edges from each $A_i$ (\cref{ln:msf-delete:ri}) to the
  initialization of $A_i$. Finally, we call \tsc{BatchInsert} on
  $\bigcup_{i=0}^{2\log n} R_i$ (\cref{ln:msf-delete:ins}) where $\abs{R_i} \le k$
  for $O(k\log^6 n)$ expected amortized work.

  \emph{Span}: The span of \tsc{Update} (\cref{alg:msf-update}) is $O(\log^2 n)$ w.h.p.
  Collapsing $A_i$ into $U$ on \cref{ln:msf-update:u} summed over all $O(\log
  n)$ iterations takes $O(\log n \log^* n)$ span w.h.p.\ using a parallel dictionary.
  Then the operations on $T_i$ and $A_i$ in
  \crefrange{ln:msf-update:tdel}{ln:msf-update:a} all take $O(\log^2 n)$ span
  w.h.p.\ and are only performed on one value of $i$.

  The span of \tsc{Insert} (\cref{alg:msf-insert}) is also $O(\log^2 n)$ w.h.p.
  The span is dominated by the RC tree updates
  (\crefrange{ln:msf-insert:del}{ln:msf-insert:ins}) and \tsc{Update}
  (\cref{ln:msf-insert:up}), each of which take $O(\log^2 n)$ span w.h.p.
  The span of \tsc{Delete} (\cref{alg:msf-delete}) is $O(\log^3 n \log k)$
  w.h.p. The span is dominated by deleting up to $k$ edges from $A_i$
  for each $i$ (\cref{ln:msf-delete:ri}) for a cost of $O(\log^3 n \log k)$ span
  w.h.p.\ according to \cref{thm:dec-msf}.

  \emph{Space}: The global tree $F$ and each local RC tree $T_i$ takes $O(n)$ space
  w.h.p., and each pair of local buffers $(B_{D, i}, B_{I, i})$ takes $O(n)$
  space since they store the symmetric difference between the trees $F$ and $T_i$.
  Each $A_i$ has $O(\min\{2^i, n\})$ vertices and
  $O(2^i)$ edges, and the initialization strategy for the local graphs in
  \tsc{Update} leaves $A_i$ empty for $i > \log m$.
  Applying \cref{thm:dec-msf} on each $A_i$ gives a total space usage of
  $
    \sum_{i=1}^{\log m} O(2^i + \min\set{2^i, n} \log^3 n)
    = O(m + \min\set{m, n\log n} \log^3 n)
  $.
\end{proof}

\subsection{Example}\label{sec:dyn-msf-example}

\providecommand{\msfExImgSize}{0.24}
\providecommand{\msgExImg}[1]{
  \raisebox{-.5\height}{
    \includegraphics[scale=\msfExImgSize]{#1}
  } 
  \vspace{.2mm}
}

\providecommand{\msfExEmpty}{\msgExImg{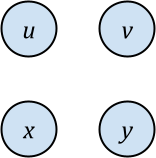}}
\providecommand{\msfExInsA}{\msgExImg{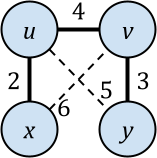}}
\providecommand{\msfExInsAF}{\msgExImg{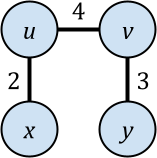}}
\providecommand{\msfExInsB}{\msgExImg{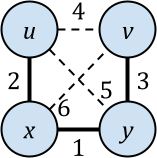}}
\providecommand{\msfExInsBF}{\msgExImg{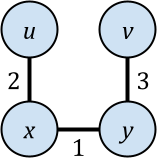}}
\providecommand{\msfExInsBLocalA}{\msgExImg{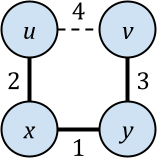}}
\providecommand{\msfExInsBLocalACompress}{\msgExImg{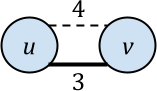}}
\providecommand{\msfExDelBef}{\msgExImg{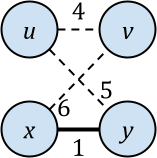}}
\providecommand{\msfExDelBefLocalA}{\msgExImg{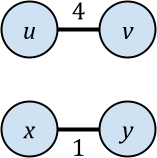}}
\providecommand{\msfExDelBefLocalACompress}{\msgExImg{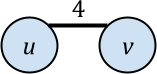}}
\providecommand{\msfExDelBefLocalB}{\msgExImg{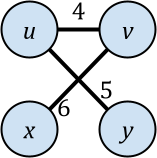}}
\providecommand{\msfExDel}{\msgExImg{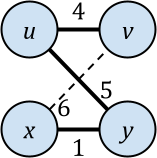}}
\providecommand{\msfExDelF}{\msgExImg{img/dyn-msf-example/del-F.png}}
\providecommand{\msfExDelLocalACompress}{\msgExImg{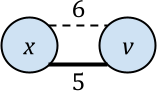}}

\begin{table*}[t]\centering
  \figfont
  \begin{tabular}{lcccccccc}
    \toprule
    Operation & $G$  & $A_0$ uncompressed & $A_0$ compressed & $T_0$ & \makecell{$B_{D,0}$ \\ $B_{I,0}$}
    & $A_1$ uncompressed & $T_1$ & \makecell{$B_{D,1}$ \\ $B_{I,1}$} \\ 
    
    \hline \noalign{\vskip 1mm}  
    Initialize & 
      \msfExEmpty & \msfExEmpty & (empty)   & \msfExEmpty & \makecell{$\set{}$\\~\\$\set{}$} 
      & \msfExEmpty & \msfExEmpty & \makecell{$\set{}$\\~\\$\set{}$} \\ 
    \hline \noalign{\vskip 1mm}
    1st insert & 
      \msfExInsA & & &  &  \makecell{$\set{}$\\~\\$\{\set{u, v},$\\$\set{u, x}, \set{v, y}\}$}
      & \msfExInsA & \msfExInsAF &  \makecell{$\set{}$\\~\\$\set{}$} \\ 
    \hline \noalign{\vskip 1mm}
    2nd insert & 
      \msfExInsB & \msfExInsBLocalA & \msfExInsBLocalACompress & \msfExInsBF 
      & \makecell{$\set{}$\\~\\$\set{}$}
      & & & \makecell{$\set{\set{u,v}}$\\~\\$\set{\set{x,y}}$}
    \\ 
    \hline \noalign{\vskip 1mm}
    \makecell[l]{Delete (before\\ insertion\\ subroutine call)} &
      \msfExDelBef & \msfExDelBefLocalA & \msfExDelBefLocalACompress &
      & \makecell{$\set{\set{u, x}, \set{v, y}}$\\~\\$\set{}$}
      & \msfExDelBefLocalB & 
      & \makecell{$\{\set{u,v},$\\$ \set{u, x}, \set{v, y}\}$\\~\\$\set{\set{x,y}}$}
    \\
    \noalign{\vskip 2mm}
    \makecell[l]{Delete (after\\ insertion\\ subroutine call)} &
      \msfExDel & \msfExDel & \msfExDelLocalACompress & \msfExDelF 
        & \makecell{$\set{}$\\~\\$\set{}$} & & 
        & 
        \makecell{$\{\set{u,v},$\\$ \set{u, x}, \set{v, y}\}$\\~\\$\{\set{u, v},$\\$\set{u, y} \set{x,y}\}$}
    \\ \bottomrule
  \end{tabular}
  \caption{This table walks through an example of the batch-dynamic MSF algorithm.
  The global graph $G$ has four vertices. We first insert five edges into it, then insert one edge, and finally delete two edges.
  The bolded edges in $G$ form the MSF (the global tree) $F$. 
  The local graphs are $A_0$ and $A_1$, and the bolded edges within are the local tree edges.
  For $i \in \set{0, 1}$, the tree $T_i$ is the RC tree for $A_i$. 
  The buffers $B_{D,i}$ and $B_{I, 1}$ represent the difference between $T_i$ and $F$, i.e., 
  $(T_i \setminus  B_{D, i}) \cup B_{I, i} = F$. For brevity, we omit listing the weights of edges 
  in the buffers.
  The compressed form of $A_1$ is the same as its uncompressed form throughout this example (except
  at initialization when the compressed form of $A_1$ is completely empty), and so we omit illustrating it.
  If a cell in the table is blank, that means it is the same as the cell in the row above.
  }
  \label{tab:dyn-msf-example}
\end{table*}

\Cref{tab:dyn-msf-example} displays an example of how the local graphs 
change as the global graph $G$ changes. Though the algorithm does not actually store
the non-tree edges of $G$ or the uncompressed forms of the local graphs $A_0$ and $A_1$, they
are displayed in the table for clarity.

In the first row of the table, we initialize a graph with four vertices $u$, $v$, $x$, and $y$.

In the second row of the table, we call \tsc{BatchInsert} on the five edges 
 $(\set{u,v}, 4)$,
 $(\set{u,x}, 2)$,
 $(\set{v,y}, 3)$,
 $(\set{u,y}, 5)$, and
 $ (\set{v,x}, 6)$.
\tsc{BatchInsert} then invokes \tsc{Update} on the two non-tree edges
$\set{u,y}$ and $\set{v,x}$. Placing these edges in
$A_0$ would violate the edge-count invariant that $A_0$ has at most 
$2^0 = 1$ local non-tree edges, but we can place them in $A_1$ because it can have
$2^1 = 2$ non-tree edges.
The \tsc{Update} call hence initializes $A_1$ on the non-tree edges
along with the current global tree. In this case, the compressed form of $A_1$ is the
same as the uncompressed form---every vertex has an incident level-1 non-tree edge and cannot be compressed out. 

In the third row of the table, we call \tsc{BatchInsert} on the edge $(\set{x,y}, 1)$.
The edge $\set{u,v}$ becomes a non-tree edge, and \tsc{Update} is invoked on it.
This time, the edge can be placed in $A_0$ without violating the edge-count invariant.
The \tsc{Update} call initializes $A_0$ on the non-tree edge
along with the current global tree. The compressed form of $A_0$
keeps the vertices $u$ and $v$ since they are non-tree edge endpoints and splices out
vertices $x$ and $y$, replacing the path $u$--$x$--$y$--$v$ with a compressed edge $\set{u, v}$
that has the same weight as the heaviest weight in the path.

In the fourth row of the table, we call \tsc{BatchDelete} on edges $\set{u, x}$ and $\set{v, y}$.
We first delete them from local graphs $A_0$ and $A_1$. Local graph $A_0$ returns local replacement edge
$\set{u,v}$, and $A_1$ returns local replacement edges $\set{u,y}$ and $\set{v,x}$. (Deleting
$\set{u, x}$ and $\set{v, y}$, or even just deleting either of these edges individually, in the
compressed form of $A_0$ means deleting the entire compressed edge representing the path
$u$--$x$--$y$--$v$. This triggers the correct behavior of searching for a local replacement edge
that reconnects $u$'s connected component and $v$'s connected component.)

Finally, \tsc{BatchDelete} invokes \tsc{BatchInsert} to re-insert all of the local replacement edges into the global graph. Edges
$\set{u,v}$ and $\set{u,y}$ are global replacement edges that are inserted as tree edges,
whereas edge $\set{v, x}$ remains a global non-tree edge. We invoke \tsc{Update} on $\set{v, x}$,
which re-initializes $A_0$ with the non-tree edge.
\section{Dynamic graph HAC lower bounds}\label{sec:hac}

In this section, we show lower bounds on dynamic graph HAC under
edge insertions and deletions. Queries take two vertices $s$ and $t$ and a similarity
threshold $\theta$, and answer whether $s$ and $t$ are in the same cluster if we
run agglomerative clustering until all cluster similarities are strictly below $\theta$.
Such queries provide limited information, but we will show
that answering such queries is still difficult for complete linkage, weighted average
linkage, and average linkage.
We also consider a different type of query that asks how many clusters
there are if we run HAC until a threshold $\theta$---we refer to this
problem as \#HAC. All of our lower bounds hold for the special case where $s$,
$t$, and $\theta$ are fixed across all queries.

Our bounds hold even for an approximate form of graph HAC. We use the
approximation notion given by Lattanzi et al.~\cite{lattanzi2019framework}. In
$\lambda$-approximate graph HAC with $\lambda \ge 1$, at an agglomeration step
where the maximum similarity is $\mathcal{W}_\text{max}$, the clustering process may
merge any clusters with similarity at least $\mathcal{W}_\text{max}/\lambda$.
We assume $\lambda \le \on{poly}(n)$ so that $\on{poly}(\lambda)$ fits in a
constant number of words.

\subsection{Background: other dynamic lower bounds}\label{sec:avw-lower-bounds}

Abboud and Vassilevska Williams as well as Henzinger et al.\ showed lower bounds on several
dynamic problems conditional on well-known conjectures~\cite{abboud2014popular,henzinger2015unifying}.
The conjectures include the strong exponential time hypothesis (SETH)~\cite{calabro2009complexity},
triangle detection requiring greater than linear work,
3SUM requiring quadratic work~\cite{gajentaan1995class,patrascu2010towards}, 
and online Boolean matrix-vector multiplication (OMv) 
requiring cubic work~\cite{henzinger2015unifying}.

%
%


The studied dynamic problems include Chan’s subset union problem
(SubUnion)~\cite{chan2006dynamic}, subgraph connectivity (SubConn), and connected subgraph
(ConnSub). In SubUnion, given a collection of sets $X = \set{X_1,\ldots , X_t}$ and $U :=
\bigcup_i X_i$, we maintain a subcollection $S \subseteq X$ under insertions and
deletions to $S$ while answering whether $\bigcup_{X_i \in S} X_i = U$.
In SubConn, given an undirected graph, we maintain a subset of vertices $S$
under insertions and deletions to $S$ with queries answering whether query vertices
$s$ and $t$ are connected in the subgraph induced by
$S$~\cite{frigioni2000dynamically}. 
ConnSub is SubConn with the query instead being whether the subgraph is connected.


\begin{table}\centering
  \figfont
  \begin{tabular}{lllll}
  \toprule
  \multirow{2}{*}{Problem} & \multicolumn{4}{c}{Work lower bounds} \\
                          & Preprocess & Update & Query & Conjecture \\
  \hline
  \multirow{2}{*}{\shortstack[l]{SubUnion with \\ $\abs{X} = O(\log \abs{U})$}}
  & $\on{poly}(\abs{U})$ & $\abs{U}^{1 - \eps}$ & $\abs{U}^{1 - \eps}$ & SETH \\
  & & & & \\
  \hline
  \multirow{3}{*}{$st$-SubConn}
  & $m^{1+\delta-\eps}$ & $m^{\delta - \eps}$ & $m^{2\delta - \eps}$ & Triangle \\
  & $m^{4/3}$ & $m^{\alpha - \eps}$ & $m^{2/3 - \alpha - \eps}$ & 3SUM \\
  & $\poly{n}$ & $m^{1/2 - \eps}$ & $m^{1-\eps}$ & OMv \\
  \hline
  ConnSub 
  & $\on{poly}(n)$ & $n^{1-\eps}$ & $n^{1-\eps}$ & SETH  \\
  \bottomrule
  \end{tabular}
  \caption{
    The table states conditional asymptotic work lower bounds for 
    some dynamic problems~\cite{abboud2014popular,henzinger2015unifying}.
    The values of $\eps$, $\delta$, and $\alpha$ follow the definitions in \cref{tab:lower-bounds}.
    The bounds hold for partially dynamic algorithms as well. The bounds are amortized in the fully dynamic case
    and are worst-case in the partially dynamic case.
  }
  \label{tab:abboud-lower-bounds}
\end{table}

\Cref{tab:abboud-lower-bounds} lists the lower bounds for these problems.
The SubUnion bounds hold for the special case where $\abs{X} = O(\log \abs{U})$, and 
the SubConn bounds hold for the special case of $st$-SubConn where $s$ and $t$ are fixed across all queries.

\subsection{Statement of HAC lower bounds}\label{sec:hac-state-lower-bounds}

The following theorem reduces SubConn to HAC and ConnSub to \#HAC.
It implies that the existing conditional lower bounds for $st$-SubConn
apply directly to HAC under complete linkage or weighted average linkage.
Similarly, the SETH gives the
same lower bounds to \#HAC as it does to ConnSub. We note that existing constructions reducing from triangle detection~\cite{abboud2014popular} and OMv~\cite{henzinger2015unifying} to SubConn also work when reducing to ConnSub, and
so the conditional lower bounds for HAC based on triangle detection and OMv hardness apply to \#HAC too.

\begin{theorem}\label{thm:subconn-reduces-to-hac}
  Suppose for some constant $c > 0$ that we can solve dynamic / incremental / decremental $O(n^c)$-approximate
  HAC under complete linkage or weighted average linkage in
   $p(m,n)$ preprocessing work, $u(m,n)$ update work, and
   $q(m,n)$ query work. Then we can solve dynamic / decremental / incremental
   SubConn with $\tilde{O}(m) + p(\tilde{O}(m),\tilde{O}(n))$ processing
   work, $u(\tilde{O}(m),\tilde{O}(n))$ update
   work, and $q(\tilde{O}(m),\tilde{O}(n))$ query work.
  The same relationship is also true between \#HAC and ConnSub.
\end{theorem}

In \aref{lem:subconn-to-hac-upgma}, we also give a reduction from SubConn to average-linkage HAC. 
Due to the large size of the HAC instance resulting from the reduction, however,
the only lower bound the reduction gives is
$\Omega(n^{1/6 - o(1)})$ update work and
$\Omega(n^{1/3 - o(1)})$ query work conditional on OMv hardness.
Meanwhile, \#HAC has $\Omega(m^{1/2 - o(1)})$ update work and
$\Omega(m^{1 - o(1)})$ query work conditional on OMv hardness\noconf{ (\cref{thm:omv-to-hac-upgma-count})} and 
also has lower bounds conditional on triangle detection hardness\noconf{ (\cref{thm:hac-upgma-count-triangle-bound})}.

Finally, the following theorem states $\Omega(n^{1-o(1)})$ dynamic
HAC lower bounds conditional on SETH. 
The lower bounds for complete-linkage and weighted-average-linkage \#HAC come from
\cref{thm:subconn-reduces-to-hac}. The remaining bounds come from
reducing SubUnion to HAC and applying existing lower bounds on SubUnion.

\begin{theorem}\label{thm:hac-seth-lower-bounds}
  Suppose that for some $\eps > 0$ we can solve one of the following problems
  with $\on{poly}(n)$ preprocessing work:
  \begin{itemize}[topsep=1pt,itemsep=0pt,parsep=0pt,leftmargin=10pt]
    \item for some constant $c$, fully dynamic $O(n^c)$-approximate HAC or
      $O(n^c)$-approximate \#HAC under complete linkage or weighted average
      linkage with $O(n^{1-\eps})$ amortized update and query work, 
    \item  the above problem in an incremental or decremental setting with
      $O(n^{1-\eps})$ worst-case update and query work,
    \item for some $c \in [0, 1/5)$, fully dynamic $O(n^c)$-approximate average-linkage HAC
      with $O(n^{1-5c-\eps})$ amortized update and query work,
    \item for some $c \in [0, 1/8)$, incremental or decremental $O(n^c)$-approximate average-linkage HAC
      with $O(n^{(1-8c)/2-\eps})$ worst-case update and query work,
    \item for some $c \in [0, 1)$, fully dynamic $O(n^c)$-approximate average-linkage \#HAC
      with $O(n^{1-c-\eps})$ amortized update and query work,
    \item for some $c \in [0, 1/2)$, incremental or decremental $O(n^c)$-approximate average-linkage \#HAC
      with $O(n^{1-2c-\eps})$ worst-case update and query work.
  \end{itemize}
  Then the SETH is false.
\end{theorem}

\subsection{Proof of bounds}\label{sec:hac-prove-lower-bounds}

As examples, we will show two reductions that generate some of our lower bounds.
We defer the remaining reductions to \aref{app:hac-lower-bounds}.

Our first example reduces SubConn to weighted-average-linkage HAC.
\noconf{\Cref{lem:subconn-to-hac-complete} gives a similar reduction for complete linkage.}

\begin{figure}
  \centering
  \begin{subfigure}{0.3\columnwidth}
        \centering
        \includegraphics[scale=0.3]{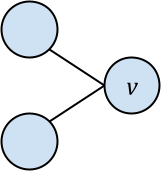}
        \caption{A vertex $v$ with two neighbors in an SubConn instance.}
    \end{subfigure}\hfill%
    \begin{subfigure}{0.65\columnwidth}
        \centering
        \includegraphics[scale=0.3]{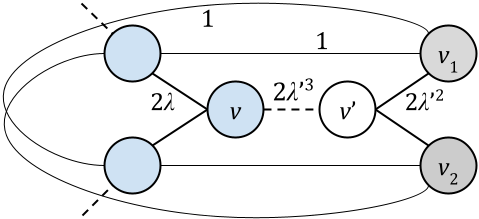}
        \caption{
          The corresponding HAC instance adds
          a star graph with center $v'$ and several leaves.
        }
    \end{subfigure}
  \caption{
    The figure displays the extra vertices added for a particular vertex $v$
    in \cref{lem:subconn-to-hac-wpgma}'s reduction from SubConn to weighted-average-linkage HAC.
  }
  \label{fig:subconn-to-hac-wpgma}
\end{figure}

\begin{lem}\label{lem:subconn-to-hac-wpgma}[Part of \cref{thm:subconn-reduces-to-hac}]
  Let $\lambda \in [1, \on{poly}(n)]$. Suppose we can solve dynamic / incremental / decremental $\lambda$-approximate
 weighted-average-linkage HAC in
 $p(m,n)$ preprocessing work, $u(m,n)$ update work, and
 $q(m,n)$ query work. Then, letting $m' = m(1 + \log \lambda)$ and $n' = n(1 +
 \log \lambda)$, we can solve dynamic / decremental / incremental SubConn
 with $O(m') + p(O(m'),O(n'))$ processing work, $u(O(m'),O(n'))$ update work,
 and $q(O(m'),O(n'))$ query work. The same relationship is also true between \#HAC and ConnSub.
\end{lem}
\begin{proof}
  Suppose we are given an unweighted graph $G = (V, E)$ and we want to solve
  SubConn or ConnSub, maintaining some dynamic subset of vertices $S$. Set $\theta =
  2\lambda$. Define $\ell = \ceil{\log(2\lambda)} = O(1 + \log\lambda)$ and
  $\lambda'=\lambda+1$.

  \emph{Preprocessing}: Construct a new, weighted graph $G'$ by copying $G$ and giving
  every edge a weight of $2\lambda$. For every $v \in V$, add a star graph to $G'$
  consisting of a center vertex $v'$ and $\ell$ leaves $v_1, \ldots, v_{\ell}$ with
  weight-$2\lambda'^2$ edges. For each $v_i$, create weight-$1$ edges to each
  vertex in $N(v)$. Connect $v$ to $v'$ with weight $2\lambda'^3$ if $v \not\in
  S$. \Cref{fig:subconn-to-hac-wpgma} illustrates this construction.
  The graph $G'$ has $O(n(1 + \log(\lambda)))$ vertices and $O(m(1 + \log(\lambda)))$ edges.
  Initialize dynamic HAC on $G'$ with $\theta = 2\lambda$.

  \emph{Update}: Simulate adding or removing a vertex $v$ in $S$ by
  removing or adding the weight-$2\lambda'^3$ edge $\set{v, v'}$.

  \emph{Query}: If we are reducing SubConn to HAC, return
  whether $s$ and $t$ are in the same cluster given similarity threshold
  $\theta$. If we are reducing ConnSub to \#HAC, then return whether the
  number of clusters is $\abs{V \setminus S} + 1$.

  \emph{Correctness}: Consider running HAC until similarity threshold $\theta =
  2\lambda$. Due to the preprocessing and the update strategy, every $v \in V \setminus S$
  has a weight-$2\lambda'^3$ edge to $v'$. These edges merge first,
  and then all of the
  weight-$2\lambda'^2$ edges merge. This puts each $v \in V \setminus S$ in a
  cluster $\set{v, v', v_1, v_2, \ldots, v_\ell}$, where the incident edges connect
  to $N(v)$ with weight $1 + (2\lambda - 1)2^{-\ell} < 2$. To see
  why the weight is $1 + (2\lambda - 1)2^{-\ell}$,
  consider $v
  \in V \setminus S$ and suppose without loss of generality that the
  weight-$2\lambda'^2$ edges for $v'$ merge in the order $v_1, v_2, \ldots,
  v_\ell$. Inductively, for $i = 0, 1, \ldots, \ell$, the incident edges on
  $v'$'s cluster have weight $1 + (2\lambda - 1)2^{-i}$ after $v_i$ merges with
  $v'$'s cluster, where $v_0 = v$. The base case is prior to merging, where the weight is $1+(2\lambda-1)2^0=2\lambda$, which is correct by construction.

  Therefore, the clusters for each $v \in V \setminus S$ do not participate in
  any more merges when clustering until threshold $\theta$. The remaining
  vertices cluster into their connected components in the subgraph induced by
  $S$, and a HAC query answers SubConn correctly. A \#HAC query will give 
  the number of connected components in the subgraph plus $\abs{V \setminus S}$ (one cluster for each
  $v \in V \setminus S$).
\end{proof}

Our second example reduces SubUnion to average linkage HAC. 
\noconf{\Cref{lem:subunion-to-hac-complete,lem:subunion-to-hac-wpgma,lem:subunion-to-hac-upgma-count}
give similar reductions to complete linkage HAC, weighted average linkage HAC, and average linkage \#HAC.}

\begin{figure}
  \includegraphics[scale=0.3]{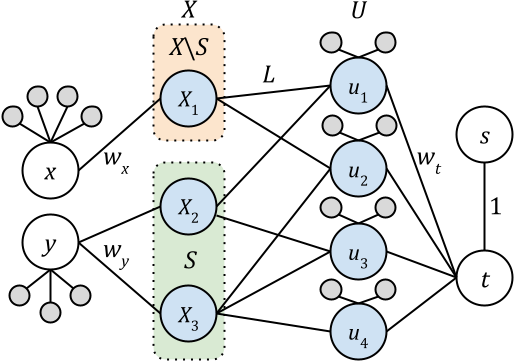}
  \caption{
  The figure illustrates the HAC instance constructed in 
  \cref{lem:subunion-to-hac-upgma}'s reduction from SubUnion
  to average-linkage HAC. The SubUnion instance in this example has
  $X = \set{X_1 = \set{u_1, u_2}, X_2 = \set{u_1, u_3}, X_3 = \set{u_2, u_3, u_4}}$
  and currently has $S = \set{X_2, X_3}$.
  We reduce the number of stars' leaves (gray vertices) displayed for cleanliness
  (e.g., $x$ should actually have dozens of leaves).
  }
  \label{fig:subunion-to-hac-upgma}
\end{figure}

\begin{lem}[Part of \cref{thm:hac-seth-lower-bounds}]\label{lem:subunion-to-hac-upgma}
  Let $\lambda \in [1, \on{poly}(n)]$.
  Suppose we can solve dynamic $\lambda$-approximate
  average-linkage HAC in
  $p(m,n)$ preprocessing work, $u(m,n)$ update work, and
  $q(m,n)$ query work. Then we can solve dynamic
  SubUnion with with $p(m', n')$ processing work, $u(m',n')$ update
  work, and $q(m',n')$ query work where $m'$ and $n'$ are
  $O(\lambda^5\abs{U}\abs{X} + \lambda^2\abs{X}^2)$.

  If we can solve incremental / decremental $\lambda$-approximate average-linkage
  HAC, then then the bounds hold for decremental / incremental
  SubUnion with $m'$ and $n'$ being $O(\lambda^{8}\abs{U}^2\abs{X} +
  \lambda^5\abs{U}\abs{X}^2 + \lambda^2\abs{X}^3)$.
\end{lem}
\begin{proof}
  Suppose we are given a SubUnion instance $(X, U)$ with subset $S \subseteq X$.
  We focus on the case where we have a fully dynamic algorithm and defer the
  partially dynamic case to \aref{app:hac-lower-bounds}. 
  Define $\theta = 1$ as well as the following constants:
  \begin{align*}
    w_t &= (\lambda + 1)\lambda = O(\lambda^2), \\
    \ell_y &= \lambda w_t\abs{U} = O(\lambda^3 \abs{U}), \\
    L &= (\lambda + 1)^2\lambda(\ell_y + 1 + \abs{X} + \lambda\abs{U}) = O(\lambda^6\abs{U} + \lambda^3\abs{X}), \\
    \ell_x &= \abs{X}L/\lambda = O(\lambda^5\abs{U}\abs{X} + \lambda^2\abs{X}^2), \\
    w_y &= (\ell_y + \abs{X})L + 1 = O(\lambda^9\abs{U}^2 + \lambda^6\abs{U}\abs{X} + \lambda^3\abs{X}^2), \\
    w_x &= \lambda(\ell_x + \abs{X})L + 1 = O(\lambda^{12}\abs{U}^2\abs{X} + \lambda^{9}\abs{U}\abs{X}^2 + \lambda^6\abs{X}^3)
  .\end{align*}

  \emph{Preprocessing}: \Cref{fig:subunion-to-hac-upgma} illustrates the graph that we will construct.
  Create a graph $G$ with a vertex representing each $X_i
  \in X$, a vertex representing each $u \in U$, and a weight-$L$ edge $\set{X_i,
  u}$ for each $u \in X_i$ for each $X_i \in X$. Make each $u \in U$ a center of
  a star graph with $\lambda - 1$ leaves connected with weight $w_x$.
  Add a star graph with center $y$ and $\ell_y$ leaves connected to the center
  with weight $w_y$. Add a weight-$w_y$ edge from $y$ to each  $X_i \in
  S$. Add another star graph with center $x$ and $\ell_x$ leaves connected to
  the center with weight $w_x$. Add a weight-$w_x$ edge from $x$ to each 
  each $X_i \in X \setminus S$. Add two more vertices $s$ and $t$ with a
  weight-$1$ edge $\set{s, t}$, and add a weight-$w_t$ edge $\set{t, u}$ for
  each $u \in U$. This graph has $O(\ell_x)$ vertices and edges. Initialize
  HAC on this graph.

  \emph{Update}: Simulate adding $X_i$ to $S$ by adding a weight-$w_y$ edge from $X_i$
  to $y$ and removing the weight-$w_x$ from $X_i$ to $x$. Similarly, simulate removing
  $X_i$ from $S$ by removing edge $\set{X_i, y}$ and adding edge $\set{X_i, x}$.

  \emph{Query}: Query whether $s$ and $t$ are in the same cluster when
  performing HAC until to threshold $\theta = 1$. If yes, then return that $S$ covers
  $U$; otherwise, return that $S$ does not cover $U$.

  \emph{Correctness}: Consider running HAC until threshold $\theta$. The weights
  $w_x$ and $w_y$ are so large that all weight-$w_x$ and weight-$w_y$ edges
  merge before any weight-$L$ edges. Let $C_x$ denote $x$'s cluster (containing
  $X \setminus S$) and $C_y$ denote $y$'s cluster (containing $S$). From this
  point onwards, the similarity between $C_x$ to another cluster $C$ is 
  bounded above by considering the worst case where every $X_i$ in $C_x$ has a
  weight-$L$ edge to every $u \in U$ contained in $C$:
  \begin{align*}
    \frac{\abs{X\setminus S}\abs{C \cap U}L}{\smallabs{C_x}\abs{C}}
    \le
    \frac{\abs{X\setminus S}L}{\smallabs{C_x}\lambda}
    <
    \frac{\abs{X}L}{\ell_x\lambda}
    = 1
  ,\end{align*}
  where the first inequality used the fact that $\lambda\abs{C \cap U} \le
  \abs{C}$ due to each $u \in U$ being in a star of size $\lambda$.
  Hence, $C_x$ does not participate in any more merges until near the end of the
  agglomeration process, at which point those merges will not affect correctness.
  On the other hand, consider $C_y$, which has weight-$L$ edges connecting it to every $u
  \in U$ covered by $S$. Each $u \in U$, until it merges with $C_y$, is in a size-$\lambda$ cluster
  consisting only of the star centered on $u$. The
  similarity between $C_y$ and some $u \in U$ covered by $S$ and not yet merged
  with $C_y$ is always at least 
  \begin{align}\label{eq:subunion-to-hac-upgma:cy}
    \frac{L}{(\ell_y + 1 + \abs{X} + \lambda\abs{U})\lambda} = (\lambda + 1)^2
  .\end{align}
  In comparison, the similarity between $t$ and another adjacent cluster $C$ is $1$ if $C =
  \set{s}$ and is otherwise at most
  \begin{align}\label{eq:subunion-to-hac-upgma:t}
    \frac{w_t\abs{C \cap U}}{\abs{C}} \le \frac{w_t}{\lambda} = \lambda + 1
  \end{align}
  where the first inequality again uses the inequality $\lambda\abs{C
  \cap U} \le \abs{C}$. Comparing \cref{eq:subunion-to-hac-upgma:cy} to \cref{eq:subunion-to-hac-upgma:t}
  shows that $C_y$ merges with all $u \in U$ covered by $S$ before $t$ merges with
  anything.

  Now consider what $t$ merges with.  In the case where $S$ does not cover all of $U$,
  inequality~\eqref{eq:subunion-to-hac-upgma:t} is tight for every
  cluster $C$ representing an uncovered $u \in U$. In particular, the similarity
  between $t$ and such a cluster is $\lambda$ times greater than the
  similarity between $t$ and $s$ or between $C_x$ and any cluster.  Therefore,
  $t$ merges with some uncovered $u$'s star rather than merging with $s$ and is in a
  cluster of size $\lambda + 1$. The similarity between $t$'s
  cluster and $s$ falls to $1/(\lambda + 1)$, and $t$ and $s$ never merge.  Hence
  a query returns the correct result in this case.

  In the case where $S$ covers all of $U$, the only adjacent clusters to $t$ are
  $\set{s}$ and $C_y$. The similarity between $t$ and $C_y$ is 
  \begin{align*}
    \frac{w_t\abs{U}}{\ell_y + 1 + \abs{S} + \lambda\abs{U}}
    < \frac{w_t\abs{U}}{\ell_y}
    = 1/\lambda
  ,\end{align*}
  which is $\lambda$ times less than the similarity between $t$ and $s$.
  Hence $t$ and $s$ merge, and a query returns the correct result in this case
  too.
\end{proof}

To turn \cref{lem:subunion-to-hac-upgma}\noconf{ (and
\cref{lem:subunion-to-hac-upgma-count})} into the lower bounds in
\cref{thm:hac-seth-lower-bounds}, we need the following lemma.

\begin{lem}\label{lem:subunion-to-approx}
  Let $a$ and $b$ be constants. Let $\mathcal{P}$ be some dynamic problem. Suppose that given a dynamic / incremental /
  decremental SubUnion instance with
  $\abs{X} = O(\log \abs{U})$, for any value of $\lambda \ge 1$, we 
  can solve the instance by efficiently converting it to an instance of 
  $\lambda$-approximate $\mathcal{P}$ of size $n' = \tilde{O}(\lambda^a\abs{U}^b)$. 
  Assuming the SETH holds, for any $c < 1/a$, the update or query work of an $O(n^c)$-approximate
  $\mathcal{P}$ algorithm with $\poly{n}$ preprocessing work is $\Omega(n^{(1-ac)/b - o(1)})$ 
  amortized / worst-case / worst-case.
\end{lem}
\begin{proof}
  Set $\lambda =
  \tilde{\Theta}(\abs{U}^{bc/(1-ac)})$ so that
  $
    n' = \tilde{O}(\lambda^a \abs{U}^b) =
    \tilde{O}(\abs{U}^{abc/(1-ac)}\abs{U}^b) = \tilde{O}(\abs{U}^{abc/(1-ac)}\abs{U}^{(b-abc)/(1-ac)}) 
    = \break  \tilde{O}(\abs{U}^{b/(1-ac)})
  $ 
  and $\lambda \ge n'^c$. Solve the SubUnion instance by
  generating a $\lambda$-approximate instance of $\mathcal{P}$ of size $n'$ and running an
  $O(n'^c)$-approximate algorithm for $\mathcal{P}$. The
  update and query time for the algorithm cannot both be
  $O(n'^{(1-ac)/b - \Omega(1)}) = \tilde{O}(\abs{U}^{1-\Omega(1)})$ because
  such a work bound for SubUnion is impossible if the SETH holds.
\end{proof}

For example, \cref{lem:subunion-to-hac-upgma} shows that SubUnion
with $\abs{X} = O(\log \abs{U})$
can be solved by running dynamic $\lambda$-approximate average-linkage HAC on a graph with
$\tilde{O}(\lambda^5\abs{U})$ vertices and edges. Setting
$(a, b) = (5, 1)$ in the lemma above gives the conditional lower
bound on fully dynamic average-linkage HAC stated in
\cref{thm:hac-seth-lower-bounds}.
\section{Conclusion}

In this paper, we gave a fully dynamic MSF algorithm that processes a
batch of $k$ updates in $O(k \log^6 n)$ expected amortized work and $O(\log^3 n \log k)$ span w.h.p.
This gives a batch-dynamic algorithm that can answer queries about single-linkage graph HAC clusters.
We also showed that graph HAC requires polynomial query or update time for
other common linkage functions unless we can break long-standing computational
complexity conjectures. This suggests that future work on dynamic HAC algorithms
for these linkage functions may wish to avoid targeting worst-case inputs.

For future work, it would be desirable to reduce the running time of the MSF
algorithm further. Can we match the $O(k \log^4 n)$ work of the
sequential HDT MSF algorithm? Can we match the $O(\log^3 n)$ span of
of Acar et al.'s best dynamic connectivity bounds? It would also be interesting to design practical implementations of our MSF algorithm.  
Finally, it would be interesting to  find restricted input classes on which 
we can break the lower bounds shown in this paper.

\begin{acks} 
We thank the reviewers for their helpful feedback.
This research was supported by DOE
Early Career Award
\#DE-SC0018947, NSF CAREER Award 
\#CCF-1845763, NSF Award \#CCF-2103483, Google Faculty Research Award, Google Research Scholar Award, FinTech@CSAIL Initiative, DARPA SDH Award \#HR0011-18-3-0007, and Applications Driving Architectures (ADA) Research Center, a JUMP Center co-sponsored by SRC and DARPA.
\end{acks}

\bibliographystyle{ACM-Reference-Format}
\bibliography{refs}

\noconf{
\appendix
\section{Relative quantile summaries}\label{app:quantile-proofs}

\subsection{Description}\label{app:quantile-proofs-seq}

In this section, we describe Zhang and Wang's mergeable relative quantile summary and give
proofs of correctness since Zhang and Wang's paper omits several
proofs~\cite{zhang2007efficient}. Some details are changed in the quantile
summaries to make our proofs work.

The quantile summaries discussed here are all of a particular form. Each summary
$Q = \set{q_1, q_2, \ldots, q_{\abs{Q}}}$ of a set $S$
is a sorted subset of $S$ with $q_1 = \min (S)$ and $q_{\abs{Q}} = \max (S)$. The subset
is simply stored as a length-$\abs{Q}$ array. For each $q \in
Q$, we maintain two integers $\rmin(q;Q)$ and $\rmax(q;Q)$ bounding the rank
of $q$ in $S$, i.e., $\rank(q;S) \in [\rmin(q;Q), \rmax(q;Q)]$. (We omit
the second argument of $\rmin(\cdot;\cdot)$ and $\rmax(\cdot;\cdot)$ when it is
clear from context.) We always maintain the minimum and maximum ranks exactly:
$\rmin$ and $\rmax$ are $1$ for $q_1$ and are $\abs{S}$ for $q_{\abs{Q}}$.  We can
assume that $\rmin$ and $\rmax$ are each strictly increasing with respect to their first 
arguments---it is easy to adjust them to be strictly increasing if not. 
The elements of $Q$, the values $\rmin(\cdot;Q)$, and the values $\rmax(\cdot;Q)$ are 
each stored in a length-$\abs{Q}$ array. 

The following lemma shows that if consecutive elements in $Q$ are close
together, then $Q$ can answer the approximate quantile queries described in
\cref{sec:quantile}.
\begin{lem}\label{lem:quantile-constraint}
  Given a summary $Q$ of the above form, suppose that for all $i \in \set{1, 2, \ldots, \abs{Q} -
  1}$ that
  \begin{equation}\label{eq:quantile-constraint}
    \rmax(q_{i+1}) - \rmin(q_i) \le \max \set{\frac{2\eps\rmin(q_i)}{1-\eps},
    1}
  .\end{equation}
  Then $Q$ is a $\eps$-approximate relative quantile summary where
  queries can be answered in $O(\log \abs{Q})$ work.
\end{lem}

We will refer to quantile summaries satisfying \cref{eq:quantile-constraint} as
\emph{ZW summaries}. To prove that \cref{lem:quantile-constraint} is true, we need to give 
an algorithm for answering queries on ZW summaries. The algorithm is given in \cref{alg:quantile-query}, 
and we discuss its correctness in the proof of \cref{lem:quantile-constraint} that immediately follows.

\begin{algorithm}
  \captionsetup{font=footnotesize}
  \caption{The algorithm for answering query for rank $r$ on an $\eps$-approximate ZW summary $Q$.}
  \label{alg:quantile-query}
\begin{algorithmic}[1]
          \figfont
  \Procedure{Query}{$Q,r$}
    \State \Return Largest $q \in Q$ such that $\rmax(q; Q) \le r(1+\eps)$ \Comment Binary search
  \EndProcedure
\end{algorithmic}
\end{algorithm}

\begin{proof}[Proof of \cref{lem:quantile-constraint}]
  The proof is similar to the proof for a similar lemma for uniform quantile
  summaries by Greenwald and Khanna~\cite{greenwald2001space}.

  Suppose we are given a query rank $r \in [1, \abs{S}]$ where $S$ is the set that
  $Q$ represents.  Assuming that $[r(1-\eps), r(1+\eps)]$ contains an integer,
  we want to return some element $y$ such that $\rank(y) \in [r(1-\eps),
  r(1+\eps)]$. The strategy will be to find some $q_i \in Q$ such that
  $[\rmin(q_i),\rmax(q_i)]\subseteq [r(1-\eps), r(1+\eps)]$. Then we will be able
  to return $q_i$ as the answer since $\rank(q_i) \in [\rmin(q_i),\rmax(q_i)]$.
  
  We run \cref{alg:quantile-query}, which binary searches in $\rmax(\cdot;Q)$ for the largest
  index $i$ such that $\rmax(q_i) \le r(1+\eps)$ and returns $q_i$. Such an index 
  always exists since $\rmax(q_1) = 1 \le r < r(1+\eps)$.
  To achieve $[\rmin(q_i),\rmax(q_i)]\subseteq [r(1-\eps), r(1+\eps)]$ as desired, it
  only remains to prove that $r(1-\eps)  \le \rmin(q_i)$.

  If $i = \abs{Q}$, then $\rmin(q_i) = \abs{S} \ge r > r(1-\eps)$. Then we are done.
  If $i < \abs{Q}$, by choice of $i$, we know that $\rmax(q_{i+1}) > r(1+\eps)$.
  Assume for contradiction that $r(1-\eps) > \rmin(q_i)$. Then
  \begin{equation}\label{eq:quantile-large-gap}
    \rmax(q_{i+1}) - \rmin(q_i) > r(1+\ep) - r(1-\ep) = 2\ep r > \frac{2 \ep
    \rmin(q_i)}{1-\ep}
  .\end{equation}
  Comparing this against \cref{eq:quantile-constraint}, we must be in the case
  where $\rmax(q_{i+1}) - \rmin(q_i) \le 1$.
  Then
  \begin{align*}
    \rmin(q_i) &< r(1-\eps) < r(1+\eps) < \rmax(q_{i+1}) \le \rmin(q_i) + 1
  .\end{align*}
  This is impossible since it shows that there is no integer between $r(1-\eps)$
  and $r(1+\eps)$. Hence we have a contradiction---we indeed have
  $[\rmin(q_i),\rmax(q_i)]\subseteq [r(1-\eps), r(1+\eps)]$ and can return $q_i$
  as our answer.
\end{proof}

\begin{proof}[Proof of \cref{lem:quantile-query}]
  Queries are answered by binary search (\cref{alg:quantile-query}) 
  in $O(\log \abs{Q})$ work.
  Fetching $\min (S)$ takes $O(1)$ time because the first element of $Q$ is always $\min
  (S)$ by construction.
\end{proof}

%

\begin{lem}\label{lem:quantile-construct-seq}
  Given a set $S$ of $n$ elements, we can construct an
  $\ep$-approximate ZW summary $Q$ that summarizes $S$ and consists of $O(\log(\ep n)/\ep)$ elements
  from $S$.
\end{lem}
\begin{proof}
  For simplicity, assume $1/\ep$ is an integer. Put all elements from $S$ of rank less
  than $1/\ep$ into $Q$. Then consider the intervals $[2^{i-1}/\ep, 2^i/\ep)$
  for $i \in \set{1,\ab 2,\ab \ldots,\ab \log(\ep n)}$. For interval $[2^{i-1}/\ep,
  2^i/\ep)$, put the elements of rank $2^{i-1}/\ep + 2^{i-1}j$ into $Q$ for each $j \in
  \{0,\ab 1,\ab 2,\ab \ldots,\ab 1/\ep - 1\}$. Finally, put the maximum of $S$ into $Q$.
  For each element $q$ in $Q$, set $\rmin(q)$ and $\rmax(q)$ to be the rank
  of $q$. The size of $Q$ is $(\log(\ep n) + 1)/\ep$.

  To see that this summary is indeed $\ep$-approximate, we want to show that it
  satisfies \cref{eq:quantile-constraint}. Consider consecutive
  elements $q_k$ and $q_{k+1}$ in $Q$. If $\rank(q_k) < 1/\ep$ then
  $\rank(q_{k+1}) = \rank(q_k) + 1$, so $\rmax(q_{k+1})-\rmin(q_k) \le 1$. Otherwise,
  the rank of $q_k$ is in some interval
  $[2^{i-1}/\ep, 2^i/\ep)$, and the rank of $q_{k+1}$ is $\rank(q_k) + 2^{i-1}$.
  Then
  \begin{align*}
    \rmax(q_{k+1})-\rmin(q_k) = \rank(q_{k+1}) - \rank(q_k) = 2^{i-1} \\ \le
    \ep \rank(q_k) < \frac{2\ep\rmin(q_k)}{1-\ep}
  .\end{align*}
  In either case we have satisfied \cref{eq:quantile-constraint}.
\end{proof}

To combine two ZW summaries $Q_1$ and $Q_2$ over non-overlapping sets
$S_1$ and $S_2$ into a new summary over $S_1 \cup S_2$, we can \emph{merge}
$Q_1$ and $Q_2$. Then, if we need to shrink the space usage of the resulting summary,
we can \emph{prune} it.

\begin{lem}[\cite{zhang2007efficient}]\label{lem:quantile-merge}
  Given $\ep$-approximate ZW summaries $Q_1$ and $Q_2$ representing sets $S_1$
  and $S_2$, we can merge the summaries
  into an $\ep$-approximate ZW summary $Q$ of size $\abs{Q_1} + \abs{Q_2}$ that
  represents set $S_1 \cup S_2$.
\end{lem}
\begin{proof}
  As described by Zhang and Wang, we construct $Q$ by merging $Q_1$
  and $Q_2$ in sorted order~\cite{zhang2007efficient}. To set $\rmin(q;Q)$ and
  $\rmax(q;Q)$ for some $q \in Q$, suppose that $q \in Q_1$ (the other case is
  symmetric). Let $s$ be the largest element of $Q_2$ smaller than $q$ and let
  $t$ be the smallest element of $Q_2$ larger than $q$. Then set
  \begin{align*}
    \rmin(q;Q) &=
    \begin{cases}
      \rmin(q;Q_1) + \rmin(s; Q_2) & \text{if $s$ exists,} \\
      \rmin(q;Q_1) & \text{otherwise,}
    \end{cases} \\
    \rmax(q;Q) &=
    \begin{cases}
      \rmax(q;Q_1) + \rmax(t; Q_2) - 1 & \text{if $t$ exists,} \\
      \rmax(q;Q_1) + \rmax(s; Q_2)   & \text{otherwise.}
    \end{cases}
  \end{align*}

  This algorithm is the same as the merge algorithm Greenwald and
  Khanna described for uniform approximate quantile
  summaries, and the analysis showing that this algorithm satisfies
  \cref{eq:quantile-constraint} is essentially the same
  as Greenwald and Khanna's proof of correctness~\cite{greenwald2004power}.
\end{proof}

\begin{lem}\label{lem:quantile-prune}
  Given a positive integer $B>0$ and an $\ep$-approximate ZW summary $Q'$
  over a set of $n$ elements, we can prune it to construct a new
  $(\ep + 1/B)$-approximate ZW summary $Q$
  of size $O(B \log (n/B))$ over the same set.
\end{lem}
\begin{proof}
  First, consider the ranks $[B, n]$ by partitioning it into the intervals
  $[2^{i-1}B,2^iB)$ for each $i \in \set{1, 2, \ldots, \log(n/B)}$. In interval
  $[2^{i-1}B,2^iB)$, query for rank $2^{i-1}(B + j)$ in $Q'$ for $j
  \in \set{0, 1, \ldots, B-1}$. The query algorithm (\cref{alg:quantile-query}) 
  guarantees non-decreasing
  elements over the calls. Place each unique element in $Q$, and place $\max (Q')$
  into $Q$ as well if it has not already been inserted.

  Now consider the result of querying for $B$, giving the smallest
  element $b$ in $Q$ so far. Take all elements of $Q'$ less than
  $b$ and place them in the front of $Q$. Since the rank of $b$ is at most
  $B(1+\eps)$, there are at most $B(1+\eps) = O(B)$ such elements.

  Finally, for each $q \in Q$, set $\rmin(q; Q) := \rmin(q; Q')$ and $\rmax(q; Q) := \rmax(q; Q')$. 
  This summary $Q$ has size $O(B \log (n/B))$.

  We want to show that $Q$ satisfies \cref{eq:quantile-constraint} with error
  $\ep + 1/B$. The elements $q \in Q$ such that $q \le b$ are a prefix of $Q'$
  that satisfy \cref{eq:quantile-constraint} with error $\ep$, so they certainly
  also satisfy the constraint with error $\ep + 1/B$. For elements greater
  than $b$, consider a pair of consecutive elements $x$ and $y$ in $Q$. There is
  some $r$ and integer $i$ such that $r \in [2^{i-1}B,2^iB)$, that querying
  $r$ resulted in $x$, and that querying $r + 2^{i-1}$ resulted in $y$. The
  query implementation guarantees that $[\rmin(x), \rmax(x)]\subseteq
  [r(1-\ep), r(1+\ep)]$ and $[\rmin(y), \rmax(y)]\subseteq [(r + 2^{i-1})(1-\ep), (r
  + 2^{i-1})(1+\ep)]$. Then
  \begin{align*}
    \rmax(y) - \rmin(x) \le (r + 2^{i-1})(1+\eps) - r(1-\eps) \\ 
    = 2\eps r + 2^{i-1}(1+\eps)
    \le 2\eps r + (r/B)(1+\eps)
    < 2\eps r + 2r/B \\
    = 2(\eps + 1/B)r 
    \le \frac{ 2(\eps + 1/B)\rmin(x)}{1 - \eps}
    < \frac{2(\eps + 1/B)\rmin(x)}{1 - (\eps + 1/B)}
  ,\end{align*}
  satisfying the constraint with error $\ep + 1/B$. Hence $Q$ is a ZW
  summary with error $\ep + 1/B$.
\end{proof}

\subsection{Parallel algorithms}\label{app:quantile-proofs-par}

This section describes parallel algorithms for ZW summaries. As subroutines, we
use parallel algorithms for merging sorted arrays, computing prefix sums, and
filtering, all of which can be done in $O(m)$ work and $O(\log m)$ span over
length-$m$ arrays~\cite{blelloch2018introduction}.

\begin{proof}[Proof of \cref{lem:quantile-construct}]
  The proof of \cref{lem:quantile-construct-seq} shows that constructing a ZW summary
  just consists of finding $O(\log(\ep n)/\ep)$ elements by rank. 
\end{proof}

\begin{lem}\label{lem:quantile-merge-par}
  Given two ZW summaries with a total size of $m$, merging them
  as described in \cref{lem:quantile-merge} can be done in $O(m)$ work and
  $O(\log m)$ span.
\end{lem}
\begin{proof}
  Let $Q_1$ and $Q_2$ be the summaries to be merged.
  Following the algorithm given in the proof of \cref{lem:quantile-merge}, use a
  parallel algorithm for merging sorted arrays to combine $Q_1$ and $Q_2$ into
  an output $Q$. When merging, keep track of whether each element is from $Q_1$ or $Q_2$.

  To compute $\rmin(\cdot;Q)$ values, each element in
  $Q_1$ must determine its predecessor in $Q_2$. Create an array $A$ of size
  $\abs{Q}$. For each index $i \in \set{1,2,\ldots,\abs{Q}}$, set $A[i]$ to 1 if
  the $i$-th element of $Q$ is from $Q_2$ and to 0 otherwise. Then compute a
  prefix sum over $A$ to obtain a list of sums $S$. Now for any $i$, if the $i$-th
  element $q_i$ of $Q$ is from $Q_1$, then the $S[i]$-th entry of $Q_2$ is the
  predecessor of $q_i$. As a result, we can now compute $\rmin(q;Q)$ for all $q$
  in $Q_1$. Similar logic allows us to compute the rest of $\rmin(\cdot;Q)$ and
  $\rmax(\cdot;Q)$.

  All of this takes $O(\abs{Q_1} + \abs{Q_2})$ work and
  $O(\log(\abs{Q_1} + \abs{Q_2}))$ span.
\end{proof}
\begin{lem} \label{lem:quantile-prune-par}
  Given a ZW summary of size $m$, pruning it as described in
  \cref{lem:quantile-prune} can be done in $O(m)$ work and
  $O(\log m)$ span.
\end{lem}
\begin{proof}
  Given an $\eps$-approximate summary $Q$ of size $m$ and a parameter $B$, the
  core part of pruning is to compute $k = O(B \log(n/B))$ rank
  queries in parallel efficiently. If $k > m$ then we should just return $Q$ as
  it is already small enough.

  Otherwise, consider the set of rank queries $\set{r_1, ..., r_k}$ that we want to compute.
  A query for rank $r_i$ consists of binary searching for the
  location of $r_i(1+\eps)$ in $\rmax(\cdot;Q)$ to get the largest element
  $q^{(i)} \in
  Q$ such that $\rmax(q^{(i)}) \le r_i(1+\eps)$ (\cref{alg:quantile-query}). We can batch all the searches
  together by creating a $k$-length array $A$ containing query value $r_i(1+\eps)$ for all $i$ and
  performing a parallel merge of $A$ with $\rmax(\cdot;Q)$ to get an output
  array $M$. Then for each $r_i(1+\eps) \in M$, its predecessor in $M$ will either
  be $q^{(i)}$ or be $r_{i-1}(1+\eps)$. If it is
  $q^{(i)}$, then we can complete the corresponding query, and
  otherwise the query can be discarded since it will have the same result as the
  query on rank $r_{i-1}$. The results of the queries can be written into a new array
  $C$, and filtering all successful queries from $C$ gives the output summary.
  Merging and filtering takes $O(m)$ work and $O(\log m)$ span.
\end{proof}
\begin{proof}[Proof of \cref{lem:quantile-combine}]
  Merge $Q_1$ and $Q_2$ in parallel into an
  $\ep$-approximate summary using \cref{lem:quantile-merge-par}.
  Then, apply \cref{lem:quantile-prune-par} to prune the
  result to size $O(b \log (n/b)$ at the cost of $1/b$ additional error using
  \cref{lem:quantile-prune-par}.
\end{proof}

\section{Tighter batch-decremental MSF work bound}\label{app:msf-work-bounds}

This section describes how to get the bound of $O(\log^3 n \log (1 + n/\Delta))$ expected amortized work per edge for the
batch-decremental MSF algorithm in \cref{thm:dec-msf}, where $\Delta$ denotes the average deletion batch size.

We use two lemmas given by Acar et al.~\cite{acar2019parallel}:

\begin{lem}\label{lem:log-sum}
For any non-negative integers $n$ and $r$,
\[
  \sum_{w=0}^r 2^w \log \left(1 + \frac{n}{2^w}\right)
  =
  O\left(
      2^r \log \left(1 + \frac{n}{2^r}\right)
   \right)
\]
\end{lem}

\begin{lem}\label{lem:log-increasing}
For any $n \ge 1$, the function $x \log (1 + n/x)$ is strictly increasing with respect to $x$ for $x \ge 1$.
\end{lem}

Now we analyze the cost of searching for replacement edges in the batch-decremental MSF algorithm.

\begin{lem}\label{lem:replacement-search-cost}
The expected work of searching for replacement edges on 
one HDT level for the batch-decremental MSF algorithm is
\[
  O\left(
      k \log^2 n \log \left(1 + \frac{n}{k} \right) 
    + p \log^2 n \log \left(1 + \frac{n}{p} \right) 
  \right)
\]
where $k$ is the batch size and $p$ is the number of edges pushed down to the next HDT level during the search.
\end{lem}
\begin{proof}
  Let $\ell$ be the HDT level that we are searching.
  The work of the replacement edge search process is dominated by inserting found replacement edges into the level-$\ell$ ETT and by updating the level-$(\ell - 1)$ ETT when pushing edges. 
  By deferring pushing edges to the end of the replacement search at level $\ell$ 
  when we have found all $p$ edges that should be pushed, 
  we update the level-$(\ell - 1)$ ETT with one operation that costs $O(p \log^2 n (1 + n/p))$ expected work.
  
  Inserting edges into the level-$\ell$ ETT, on the other hand, occurs on each round within the replacement search on this level. 
  There are $O(\log k)$ rounds, and in each round, we search for a replacement edge out of every ``active'' component, 
  i.e., each component that still has incident edges to search and that has size at most $2^{\ell - 1}$. At most half as many components will be active on each subsequent round because 
  the worst case is when the replacement edges discovered in a round pair the components off. The number of edges inserted in a round is bounded by the number of active components, and
  on the first round, there are at most $2k$ active components. 
  Therefore, on round $i$, the number of edges $k_i$ inserted into the ETT is at most $k_i \le 4k2^{-i}$. The work of insertions across all $O(\log k)$ rounds is then proportional to
  \begin{align*}
    \sum_{i=1}^{O(\log k)} k_i \log^2 n \log \left(1 + \frac{n}{k_i}\right) \\
    \le
    \sum_{i=1}^{O(\log k)} 4k 2^{-i} \log^2 n \log \left(1 + \frac{n}{4k2^{-i}}\right) \\
    = O\left(4k \log^2 n \log\left(1 + \frac{n}{4k}\right)  \right)
    = O\left(k \log^2 n \log\left(1 + \frac{n}{k}\right)  \right)
  ,\end{align*}
  where \cref{lem:log-increasing} provides the inequality and \cref{lem:log-sum} provides the first equality.
  
  The expected work from inserting replacement edges and pushing edges is then
    $ O(k \log^2 n \log (1 + n/k) + p \log^2 n \log(1 + n/p))$.
\end{proof}

Similar to the above lemma minus a multiplicative factor of $O(\log^2 n)$, Acar et al.~\cite{acar2019parallel} 
prove in Lemma 4.5 of their paper that the 
expected work of searching for replacement edges on one HDT level for their ``interleaved'' connectivity algorithm 
is $O(k \log (1 + n/k) + p \log (1 + n/p))$ where $k$ is the batch size and $p$ is the number of pushed edges. 
Acar et al. then use their lemma to show in Theorem 4.7 of their paper that their interleaved connectivity algorithm achieves 
$O(\log n \log (1 + n/\Delta))$ expected amortized work per edge by summing the work of replacement searches 
across all HDT levels and all batches of deletions. We can apply the same analysis as Acar et al.'s Theorem 4.7 to get
the same work plus a multiplicative factor of $O(\log^2 n)$ for our batch-decremental MSF algorithm, i.e.,
$O(\log^3 n \log (1 + n/\Delta))$ expected amortized work per edge, as desired.
\section{Dynamic HAC dendrogram counterexamples}\label{app:dendro-counterexamples}

The canonical output for HAC is a dendrogram representing the hierarchy of clusters.
A natural question for dynamic graph HAC is whether we can maintain the dendrogram dynamically.
In this section, we show examples where an edge update causes $\Omega(n)$ changes in the dendrogram. 
Explicitly maintaining the dendrogram dynamically therefore seems inefficient in the worst case.
Here we assume a dendrogram is represented a binary tree (for each connected component of the graph) 
with labeled leaves, and changes in a dendrogram mean
changes in the parent or child pointers of a dendrogram node. 

\subsection{Single linkage}\label{app:dendro-single}

\begin{figure}
    \centering
    \begin{subfigure}{0.5\columnwidth}
      \centering
      \includegraphics[scale=0.3]{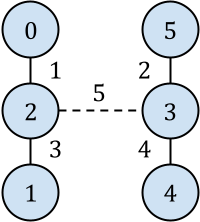}
      \caption{Graph}
    \end{subfigure}
    \begin{subfigure}{0.5\columnwidth}
        \centering
        \begin{tikzpicture}[scale=0.6]
            \node (a) at (0,0) {$0$};
            \node (b) at (1,0) {$1$};
            \node (c) at (2,0) {$2$};
            \node (d) at (3,0) {$3$};
            \node (e) at (4,0) {$4$};
            \node (f) at (5,0) {$5$};
            \node (de) at (3.5,1.5) {};
            \node (bc) at (1.5,2) {};
            \node (def) at (4.25,2.5) {};
            \node (abc) at (0.75,3) {};
            
            \draw  (d) |- (de.center);
            \draw  (e) |- (de.center);
            \draw  (b) |- (bc.center);
            \draw  (c) |- (bc.center);
            \draw  (de.center) |- (def.center);
            \draw  (f) |- (def.center);
            \draw  (bc.center) |- (abc.center);
            \draw  (a) |- (abc.center);
        \end{tikzpicture}
        \caption{Dendrogram without edge $\set{2,3}$}
    \end{subfigure}%
    \begin{subfigure}{0.5\columnwidth}
        \centering
        \begin{tikzpicture}[scale=0.6]
            \node (a) at      (0,0) {$0$};
            \node (b) at      (1,0) {$1$};
            \node (c) at      (2,0) {$2$};
            \node (d) at      (3,0) {$3$};
            \node (e) at      (4,0) {$4$};
            \node (f) at      (5,0) {$5$};
            \node (cd) at     (2.5,1) {};
            \node (cde) at    (3.25, 1.5) {};
            \node (bcde) at   (2.125, 2) {};
            \node (bcdef) at  (3.5625, 2.5) {};
            \node (abcdef) at (1.78125, 3) {};
            
            \draw  (c) |- (cd.center);
            \draw  (d) |- (cd.center);
            \draw  (cd.center) |- (cde.center);
            \draw  (e) |- (cde.center);
            \draw  (cde.center) |- (bcde.center);
            \draw  (b) |- (bcde.center);
            \draw  (bcde.center) |- (bcdef.center);
            \draw  (f) |- (bcdef.center);
            \draw  (bcdef.center) |- (abcdef.center);
            \draw  (a) |- (abcdef.center);
        \end{tikzpicture}
        \caption{Dendrogram with edge $\set{2,3}$}
    \end{subfigure}
    \caption{Example of a graph whose single-linkage HAC dendrogram
      changes a lot if an edge is added.}
    \label{fig:dendro-single}
\end{figure}

Let $n \in \N$ be even. Consider a graph consisting of 
two star graphs of $n/2$ vertices. In one star, the edge weights
are $1, 3, 5, \ldots, n - 3$, and in the other star, the edge weights
are $2, 4, 6, \ldots, n - 2$. 

Adding an edge of weight $n - 1$ between the 
centers of the two stars causes $\Theta(n)$ changes to the dendrogram. 
\cref{fig:dendro-single}
shows this graph with $n = 6$ along with its dendrograms.
Without the weight-$(n-1)$ edge, the dendrogram consists of two binary trees with long spines.
In one of the stars, consider any vertex $v$ whose incident edge's weight is $w < n-3$, and
let $u$ be the vertex whose incident edge's weight is $w + 2$, e.g., we could pick 
$v=0$ and $u=1$ in \cref{fig:dendro-single}.
In the dendrogram, $v$'s parent's other child is the parent of $u$. With the weight-$(n-1)$ edge added,
however, the two binary trees ``interleave'', and $v$'s parent's other child is no longer the parent of $u$.
Therefore, either $v$'s parent changed, $u$'s parent changed, or $v$'s parent's child changed.
This is true for any of the $\Theta(n)$ choices of $v$. Moreover, no two vertices share parents, so changes to
one vertex's parent's child do not overlap with changes to another vertex's parent's child.
This shows that the number of changes in the dendrogram caused by the addition of the 
weight-$(n-1)$ edge is indeed $\Theta(n)$.

\subsection{Complete and weighted average linkage}\label{app:dendro-wpgma}

\begin{figure}
    \centering
    \begin{subfigure}{0.5\columnwidth}
      \centering
      \includegraphics[scale=0.3]{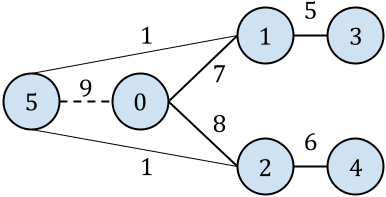}
      \caption{Graph}
    \end{subfigure}
    \begin{subfigure}{0.5\columnwidth}
        \centering
        \begin{tikzpicture}[scale=0.6]
            \node (f) at (0,0) {$5$};
            \node (a) at (1,0) {$0$};
            \node (c) at (2,0) {$2$};
            \node (b) at (3,0) {$1$};
            \node (e) at (4,0) {$4$};
            \node (d) at (5,0) {$3$};
            \node (ac) at (1.5,1) {};
            \node (abc) at (2.25,1.5) {};
            \node (abce) at (3.125,2) {};
            \node (abcde) at (4.0625,2.5) {};
            \node (abcdef) at (2.03125,3) {};
            
            \draw  (a) |- (ac.center);
            \draw  (c) |- (ac.center);
            \draw  (ac.center) |- (abc.center);
            \draw  (b) |- (abc.center);
            \draw  (abc.center) |- (abce.center);
            \draw  (e) |- (abce.center);
            \draw  (abce.center) |- (abcde.center);
            \draw  (d) |- (abcde.center);
            \draw  (abcde.center) |- (abcdef.center);
            \draw  (f) |- (abcdef.center);
        \end{tikzpicture}
        \caption{Dendrogram without edge $\set{0,5}$}
    \end{subfigure}%
    \begin{subfigure}{0.5\columnwidth}
        \centering
        \begin{tikzpicture}[scale=0.6]
            \node (f) at (0,0) {$5$};
            \node (a) at (1,0) {$0$};
            \node (c) at (2,0) {$2$};
            \node (e) at (3,0) {$4$};
            \node (b) at (4,0) {$1$};
            \node (d) at (5,0) {$3$};
            \node (af) at     (0.5,1) {};
            \node (ce) at    (2.5, 1.5) {};
            \node (bd) at   (4.5, 2) {};
            \node (acef) at   (1.5, 2.5) {};
            \node (abcdef) at (3, 3) {};
            
            \draw  (a) |- (af.center);
            \draw  (f) |- (af.center);
            \draw  (c) |- (ce.center);
            \draw  (e) |- (ce.center);
            \draw  (b) |- (bd.center);
            \draw  (d) |- (bd.center);
            \draw  (af.center) |- (acef.center);
            \draw  (ce.center) |- (acef.center);
            \draw  (acef.center) |- (abcdef.center);
            \draw  (bd.center) |- (abcdef.center);
        \end{tikzpicture}
        \caption{Dendrogram with edge $\set{0,5}$}
    \end{subfigure}
    \caption{Example of a graph whose HAC dendrogram under complete linkage or
      weighted average linkage
      changes a lot if an edge is added.}
    \label{fig:dendro-wpgma}
\end{figure}
  
Let $k \in \N$. Create a graph on $n = 2k+2$ vertices. For $i \in \set{1, 2, \ldots, k}$,
add edge $\set{0, i}$ with weight $3k + i$, edge $\set{i, k +i}$ with weight $2k + i$, and edge
$\set{i, n-1}$ with weight $1$.

Adding an edge of weight $4k + 1$ between vertices $0$ and $n-1$ 
produces $\Theta(n)$ changes in the dendrogram.
With the weight-$(4k + 1)$ edge, for all $i \in \set{1, 2, \ldots, k}$, vertices $i$ and $k+i$ 
directly merge with each other. Therefore, vertices $i$ and $k+i$ 
have the same parent in the dendrogram.
Without the edge, each vertex successively merges as a singleton cluster with vertex $0$'s cluster.
Then no vertices have the same parent in the dendrogram besides vertices $0$ and $n-1$.

\subsection{Average linkage}

\begin{figure}
    \centering
      \includegraphics[scale=0.3]{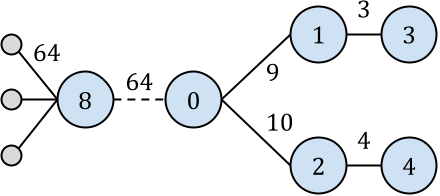}
    \caption{Example of a graph whose average-linkage HAC dendrogram an edge is added.}
    \label{fig:dendro-upgma}
\end{figure}

The construction is similar to the graph given for complete and weighted average linkage in
\cref{app:dendro-wpgma}.
Let $k \in \N$. Create a graph on $n = 4k + 1$ vertices.
For $i \in \set{1, 2, \ldots, k}$, add edge $\set{0, i}$ with weight $2k^2 + i$ and edge
$\set{i, k + i}$ with weight $k + i$. For $i \in \set{2k + 1, 2k + 2, 2k + 3, \ldots, 4k - 1}$,
add edge $\set{i, n - 1}$ with weight $8k^3$.

Adding an edge of weight $8k^3$ between vertices $0$ and $n-1$ produces $\Theta(n)$ changes to the dendrogram.
\cref{fig:dendro-upgma} shows the graph with $k = 2$.
As in \cref{app:dendro-wpgma}, with the extra edge, vertices $1$ through $2k$ merge as pairs, whereas
without the extra edge, those vertices merge as singleton clusters with vertex $0$.

\section{HAC lower bounds}\label{app:hac-lower-bounds}
This section contains deferred proofs of the lower bounds stated
in \cref{sec:hac-prove-lower-bounds}.

To complete the proof of \cref{thm:subconn-reduces-to-hac}, 
we give a reduction from SubConn to complete-linkage HAC.

\begin{figure}
  \centering
  \begin{subfigure}{0.3\columnwidth}
        \centering
        \includegraphics[scale=0.3]{img/reductions/subconn-to-hac-before.png}
        \caption{A vertex $v$ with two neighbors in an SubConn instance.}
    \end{subfigure}\hfill%
    \begin{subfigure}{0.65\columnwidth}
        \centering
        \includegraphics[scale=0.3]{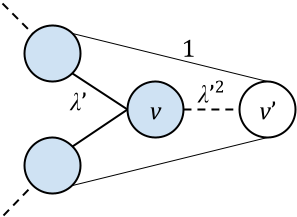}
        \caption{
          The corresponding HAC instance adds
          a vertex $v'$.
        }
    \end{subfigure}
  \caption{
    The figure displays the extra vertex and edges added for a particular vertex $v$
    in \cref{lem:subconn-to-hac-complete}'s reduction from SubConn to complete-linkage HAC.
  }
  \label{fig:subconn-to-hac-complete}
\end{figure}

\begin{lem}\label{lem:subconn-to-hac-complete}[Part of \cref{thm:subconn-reduces-to-hac}]
  Let $\lambda \in [1, \on{poly}(n)]$.
  Suppose we can solve dynamic / incremental /
  decremental $\lambda$-approximate complete-linkage HAC $p(m,n)$
  preprocessing work, $u(m,n)$ update work, and $q(m,n)$ query work. Then we can
  solve dynamic / decremental / incremental SubConn with $O(m) +
  p(O(m),O(n))$ processing work, $u(O(m),O(n))$ update work, and $q(O(m),O(n))$
  query work.  The same relationship is also true between \#HAC and ConnSub.
\end{lem}
\begin{proof}
  Suppose we are given an unweighted graph $G = (V, E)$ and we want to solve SubConn,
  maintaining some subset of vertices $S$, using HAC. Define $\lambda' = \lambda + 1$ and set
  $\theta = \lambda'$.

  \emph{Preprocessing}: Construct a new, weighted graph $G'$ by copying $G$ and giving
  every edge a weight of $\lambda'$. For every $v \in V$, create a vertex $v'$ and
  add it to $G'$.  Create weight-$1$ edges by from $v'$ to each neighbor in  $N(v)$, and then
  connect $v$ to $v'$ with weight $\lambda'^2$ if $v \in V \setminus S$.  
  \Cref{fig:subconn-to-hac-complete} illustrates this construction.
  The new graph
  $G'$ has $2n$ vertices and at most $2m + n$ edges. Initialize dynamic
  HAC on graph $G'$ with $\theta = \lambda'$

  \emph{Update}: Simulate inserting a vertex $v$ into $S$ by removing edge $\set{v,
  v'}$, and simulate deletion by adding edge $\set{v, v'}$ with weight
  $\lambda'^2$.

  \emph{Query}: If we are  reducing SubConn to HAC, return
  whether $s$ and $t$ are in the same cluster given similarity threshold
  $\theta$. If we are reducing ConnSub to \#HAC, then return whether the
  number of clusters is $\abs{V \setminus S} + 1$.

  \emph{Correctness}: Consider running HAC until similarity
  threshold $\theta$. First, all the weight-$\lambda'^2$ edges merge, leaving
  each $v \in V \setminus S$ in a cluster $\set{v, v'}$ that only has incident
  edges of weight $1$. These clusters do not participate in any more merges.
  The remaining vertices $v \in S$ cluster into their connected components in
  the subgraph induced by $S$.
\end{proof}

We can similarly reduce SubConn to average-linkage HAC, albeit with worse bounds:

\begin{lem}\label{lem:subconn-to-hac-upgma}
  Let $\lambda \in [1, \on{poly}(n)]$.
  Suppose we can solve dynamic / incremental /
  decremental $\lambda$-approximate average-linkage HAC $p(m,n)$
  preprocessing work, $u(m,n)$ update work, and $q(m,n)$ query work. 
  Then, with $m'$ and $n'$ being $O(\lambda n^3)$, we can solve dynamic /
  decremental / incremental SubConn with $O(m') + p(O(m'),O(n'))$ preprocessing
  work, $u(O(m'),O(n'))$ update work, and $q(O(m'),O(n'))$ query work.
  The same relationship is also true between \#HAC and ConnSub.
\end{lem}
\begin{proof}
  \emph{Preprocessing}: Construct a new, weighted graph $G'$ by copying $G$ and giving
  every edge a weight of $1$. Add a star graph to $G'$ consisting of a new center
  vertex $x$ and $\lambda n^3 - 1$ leaves with weight-$2\lambda^2 n^3$ edges. Add a
  weight-$2\lambda^2 n^3$ edge from $x$ to each $v \in V \setminus S$. Initialize
  dynamic HAC on $G'$ with $\theta = 4/n^2$.

  \emph{Update and query}: Simulate inserting or removing a vertex $v$ in $S$ by removing
  or adding the weight-$2\lambda^2 n^3$ edge $\set{x, v}$. Querying is the same
  as $\cref{lem:subconn-to-hac-complete}$ with $\theta = 4/n^2$.

  Correctness: The weight-$2\lambda^2 n^3$ edges have such high weight that they merge before 
  any weight-$1$ edge. After that, $x$'s cluster $C_x$ contains $V \setminus S$ and has size $\abs{V \setminus S} + \lambda n^3$.
  The similarity between $x$'s cluster $C_x$ and
  another cluster $C$ is at most
  \begin{align*}
  \frac{\abs{V \setminus S}\abs{C}}{(\abs{V \setminus S} + \lambda n^3)\abs{C}}
  < \frac{n}{2\lambda n^3}
  = \frac{1}{2\lambda n^2}
  < \frac{\theta}{\lambda}
  \end{align*}
  with the numerator on the left-hand side representing the worst case where every
  vertex of $V \setminus S$ has an edge to every vertex of $C$. 
  Therefore, $C_x$ experiences no more merges when clustering until threshold $\theta$.
  The smallest similarity between any two adjacent clusters in the subgraph induced by $S$, 
  on the other hand, is
  at least $4/n^2 = \theta$ (achieved by having two clusters of $n/2$ vertices
  each and only one edge $e \in E$ crossing their cut), so adjacent clusters will merge. The subgraph induced by $S$ hence merges into its connected components.
\end{proof}

The cubic preprocessing work in
\cref{lem:subconn-to-hac-upgma} is too high to achieve lower bounds conditional
on triangle detection hardness or 3SUM hardness, but it does allow lower bounds
conditional on online Boolean matrix-vector multiplication (OMv) hardness. 

In the OMv problem, we are
given an $n$-by-$n$ Boolean matrix $M$ and $n$ column vectors $v_1, \ldots, v_n
\in \set{0,1}^n$ one-by-one. For each $i \in \set{1, ..., n}$, we must 
output $Mv_i$ over the Boolean semiring, and we cannot see $v_j$ for any $j > i$
until we output $Mv_i$.
The \emph{OMv hardness conjecture} is that solving OMv with an error probability
of at most $1/3$ requires $\Omega(n^{3-o(1)})$ work~\cite{henzinger2015unifying}.

\begin{theorem}\label{thm:omv-to-upgma}
  Suppose for $\eps > 0$ and $c \in [0, 1)$ that we can solve
  dynamic / partially dynamic $O(n^c)$-approximate average-linkage HAC with 
  $\on{poly}(n)$ preprocessing work,
  $O(n^{(1-c)/6 - \eps})$ amortized / worst-case update work, and 
  $O(n^{(1-c)/3 - \eps})$ query work. Then the OMv hardness conjecture is 
  false.
\end{theorem}
\begin{proof}
  Fix arbitrary $c \in [0, 1)$ and suppose we had an $n'^c$-approximate (on
  inputs with $n'$ vertices) HAC algorithm matching
  the work bounds described in the theorem. Then if we were given an 
  $st$-SubConn instance with $n$ vertices and $m$ edges, 
  we can apply the construction from \cref{lem:subconn-to-hac-upgma}
  with $\lambda = \Theta(n^{3c/(1-c)})$ to get a
  $\lambda$-approximate average-linkage HAC instance 
  over a graph $G'$ with $m'$ and $n'$ being $O(\lambda n^3) = O(n^{3/(1-c)})$
  and $\lambda \ge n'^c$. Then we can solve the instance with our
  $n'^c$-approximate HAC algorithm with $\on{poly}(n)$ preprocessing work,
  $O(n^{1/2-\Omega(1)})$ update work, and 
  $O(n^{1-\Omega(1)})$ query work. Henzinger et al.\ show that such a fast
  algorithm for $st$-SubConn is impossible conditional on OMv hardness~\cite{henzinger2015unifying}.
\end{proof}

\Cref{lem:subconn-to-hac-upgma} also implies a lower bound on average-linkage
\#HAC conditional on OMv hardness since 
Henzinger et al.'s~\cite{henzinger2015unifying} construction that reduces
OMv to $st$-SubConn also works for ConnSub, but we can get a stronger lower bound
by a different reduction. Henzinger et al.\ prove the hardness of $st$-SubConn
conditional on OMv hardness by proving that it suffices to reduce
a related problem called $1$-uMv to $st$-SubConn. In the reduction, Henzinger
et al.\ construct a bipartite graph $G$ and use $st$-SubConn to solve the following
problem: ``activate'' several vertices of $G$ on demand and determine whether 
there is any edge connecting activated vertices. 
Instead of $st$-SubConn, we can use average-linkage \#HAC to solve this 
bipartite graph problem and hence get a \#HAC lower bound conditional on OMv.

Similarly, Abboud and Vassilevska Williams prove the hardness of $st$-SubConn
conditional on triangle detection hardness by constructing a bipartite graph and 
solving the same problem of determining whether an edge in the bipartite graph
connects activated vertices~\cite{abboud2014popular}.
We can then also get an average-linkage \#HAC lower bound conditional on 
triangle detection hardness. 

The following theorem states the average-linkage \#HAC lower bound conditional on OMv.

\begin{theorem}\label{thm:omv-to-hac-upgma-count}
  Suppose for $\eps > 0$ and $c \in [0, 1)$ that we can solve
  dynamic / partially dynamic $O(m^c)$-approximate average-linkage \#HAC with 
  $\on{poly}(n)$ preprocessing work,
  $O(m^{(1-c)/2 - \eps})$ amortized / worst-case update work, and
  $O(m^{1-c - \eps})$ query work. Then the OMv hardness conjecture is 
  false.
\end{theorem}
\begin{proof}
  As discussed in the text above, we 
  construct a bipartite graph following Henzinger et al.'s
  reduction from $1$-uMv to $st$-SubConn. Then we want to use \#HAC
  determine whether there is any edge connecting activated vertices in the
  bipartite graph. We omit the remainder of this proof because it is similar 
  to the proofs of \cref{thm:hac-upgma-count-triangle-bound} and
  \cref{lem:triangle-to-hac-upgma} for reducing triangle detection to \#HAC.
\end{proof}

We state the average-linkage \#HAC lower bound conditional on 
triangle detection hardness next. The \emph{triangle detection hardness conjecture}
is that there is a constant $\delta > 0$
such that any word-RAM algorithm detecting whether an $m$-edge graph has a
triangle requires $\Omega(m^{1+\delta-o(1)})$ expected work~\cite{abboud2014popular}.
(The best existing triangle detection
algorithms take
$O(\min\{m^{2\omega/(\omega+1)},n^\omega\})$ work where $O(n^\omega)$ is
the work to multiply $n$-by-$n$ matrices~\cite{itai1978finding,alon1997finding}.)

\begin{theorem}\label{thm:hac-upgma-count-triangle-bound}
  Suppose for $\delta > 0$, $\eps > 0$, and $c \in [0, 1)$ that we can solve one
  of the following problems with $O(m^{(1+\delta)(1-c)-\eps})$ preprocessing
  work, $O(m^{\delta(1-c) - \eps})$ update work, and $O(m^{2\delta(1-c) -
  \eps})$ query work:
  \begin{itemize}[topsep=1pt,itemsep=0pt,parsep=0pt,leftmargin=10pt]
    \item fully dynamic $O(m^c)$-approximate average-linkage \#HAC with the
      work bounds being amortized,
    \item incremental (on dense graphs) or decremental $O(m^c)$-approximate
      average-linkage \#HAC with the work bounds being worst-case.
  \end{itemize}
  Then the triangle detection hardness conjecture is false for this value of $\delta$.
\end{theorem}

We prove the above theorem by directly reducing
triangle detection to average-linkage \#HAC in the following lemma.

\begin{figure}
  \centering
  \begin{subfigure}{0.3\columnwidth}
        \centering
        \includegraphics[scale=0.3]{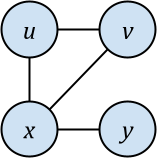}
        \caption{Triangle detection instance}
        \label{fig:triangle-to-hac-upgma-before}
    \end{subfigure}\hfill%
    \begin{subfigure}{0.65\columnwidth}
        \centering
        \includegraphics[scale=0.3]{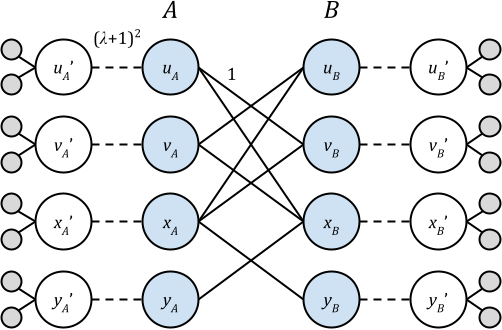}
        \caption{
          The corresponding HAC instance
        }
        \label{fig:triangle-to-hac-upgma-after}
    \end{subfigure}
  \caption{
    On the left is an example graph, and on the right is the 
    result of applying to the example graph the
    reduction from triangle detection to average-linkage \#HAC
    given by \cref{lem:triangle-to-hac-upgma}.
  }
  \label{fig:triangle-to-hac-upgma}
\end{figure}

\begin{lem} \label{lem:triangle-to-hac-upgma}
  Let $\lambda \in [1, \on{poly}(n)]$. Suppose that we can solve dynamic /
  decremental $\lambda$-approximate average-linkage
  \#HAC with preprocessing work
  $p(m, n)$, update work $u(m, n)$, and query work $q(m, n)$ amortized /
  worst-case. Let $m' = O(\lambda n + m)$ and $n' = O(\lambda n)$. Then
  triangle detection can be solved in $O(m \cdot u(m', n') + n \cdot q(m', n') + p(m',
  n'))$ work. If we have an incremental algorithm instead with worst-case
  work bounds, then triangle detection can be solved in $O(n^2 \cdot u(m', n')
  + n \cdot q(m', n') + p(m', n'))$ work.
\end{lem}
\begin{proof}
  Suppose we are given an unweighted graph $G = (V,E)$ on which we want to 
  detect a triangle.
  Following the construction in Abboud and Vassilevska Williams's reduction
  from triangle detection to $st$-SubConn~\cite{abboud2014popular}, 
  we construct an (initially unweighted) bipartite graph $G' = (V' = A \cup B, E')$
  with partitions $A$ and $B$ such that for each $v \in V$ we create $v_A \in A$
  and $v_B \in B$. For every edge $\set{u, v} \in E$, create edges $\set{u_A,
  v_B}$ and $\set{v_A, u_B}$. Abboud and Vassilevska Williams remark that given
  an arbitrary vertex $v \in V$, if we ``activate'' $u_A$ and $u_B$ for every
  neighbor $u$ of $v$, then $v$ participates in a triangle if there is an edge 
  connecting two activated vertices. Therefore, to detect whether the graph $G$  
  has a triangle, we can iterate through each
  $v \in V$, activate $v$'s neighbors, 
  check for an edge between activated vertices, and deactivate $v$'s neighbors.
  
  We focus on the case of fully dynamic \#HAC.
  To solve this bipartite graph problem using \#HAC, we 
  give every existing edge in $E'$ a weight of $1$. For each $v \in A \cup B$, add
  a star graph to $G'$ consisting of a center vertex $v'$ and $\lambda - 1$ 
  leaves with weight-$(\lambda + 1)^2$ edges. 
  Add a weight-$(\lambda + 1)^2$ edge
  $\set{v, v'}$ as well. Initialize \#HAC on $G'$
  with threshold $\theta = 1$. We can ``activate'' a vertex $v$ by deleting edge $\set{v,
  v'}$ and ``deactivate'' it by re-inserting the edge.
  The graph $G'$ has $m' = O(\lambda n + m)$ edges and $n' = O(\lambda n)$ vertices.
  \Cref{fig:triangle-to-hac-upgma} illustrates this construction with $G$ being
  \cref{fig:triangle-to-hac-upgma-before} and $G'$ being
  \cref{fig:triangle-to-hac-upgma-after}.
  
  To determine whether there is an edge between activated vertices, query for
  the number of clusters if we cluster until threshold $\theta$.
  If the number of clusters is $\abs{A} + \abs{B}$, then there is no such edge,
  otherwise there is such an edge. This is true because all deactivated vertices
  merge with their corresponding star, and then their incident edges fall below
  weight $1/\lambda$ and hence are no longer considered when clustering until
  threshold $\theta = 1$. If there are no edges between two active vertices,
  then no more merges occur and we are left with $\abs{A} + \abs{B}$ clusters.
  If there is an edge between two active vertices, at least one of them will
  merge and we will be left with fewer clusters.
  
  The number of \#HAC queries is $O(n)$ (once for each vertex in $V$), and the 
  number of \#HAC updates is $O(m)$ (once for each neighbor of each vertex 
  in $V$, i.e., twice for each edge of $E$). The total work is then
  $O(m \cdot u(m', n') + n \cdot q(m', n') + p(m', n'))$. 
  
  As Abboud and Vassilevska Williams note for the partially dynamic case,
  we cannot deactivate vertices by inserting edges
  if we are considering decremental \#HAC. Instead, we roll back the
  state of the \#HAC algorithm until the vertices are no longer activated.
  This rolling back means that we can only analyze worst-case work and not amortized
  work. For incremental \#HAC, we initialize the graph $G'$ to not have 
  the edges of the form $\set{v, v'}$ (i.e., all vertices are activated). When we
  are iterating over $v \in V$, instead of activating $v$'s neighbors, we deactivate
  its non-neighbors. This increases the number of \#HAC updates from $O(m)$ to $O(n^2)$.
\end{proof}
\begin{proof}[Proof of \cref{thm:hac-upgma-count-triangle-bound}]
  Fix arbitrary $c \in [0, 1)$ and suppose we had an $m'^c$-approximate (on
  inputs with $m'$ edges) \#HAC algorithm matching
  the work bounds described in the theorem. Then given a graph $G$ on $n$
  vertices and $m \ge n$ edges upon which we want to solve triangle detection,
  we can apply the construction from \cref{lem:triangle-to-hac-upgma} with
  $\lambda = \Theta(m^{c/(1-c)})$ to get a $\lambda$-approximate average-linkage \#HAC instance
  over a graph $G'$ with $m' = O(m^{c/(1-c)}n + m) \le O(m^{c/(1-c)}m) =
  O(m^{1/(1-c)})$ and with $\lambda = m'^c$. We apply a
  $m'^c$-approximate \#HAC algorithm to solve triangle
  detection via \cref{lem:triangle-to-hac-upgma} in $O(m \cdot u(m', n') + n \cdot q(m', n') + p(m',
  n'))$ work (for fully dynamic and decremental \#HAC; the analysis for
  incremental \#HAC swaps the $m$ factor with $n^2$ and hence only gives 
  bounds on dense graphs). Substituting in the bounds of the \#HAC
  algorithm and substituting $m' = O(m^{1/(1-c)})$ shows that we've solved
  triangle detection in
  $O(m^{1+\delta - \Omega(1)} + nm^{2\delta - \Omega(1)})$ work.  Abboud and
  Vassilevska Williams provide a lemma showing that such a fast triangle detection
  algorithm gives an $O(m^{1+\delta - \Theta(1)})$ work algorithm for triangle
  detection~\cite{abboud2014popular}.
\end{proof}

Next, we move on to reductions from SubUnion. 
We finish the proof of \cref{lem:subunion-to-hac-upgma} by reducing partially dynamic SubUnion 
to partially dynamic average-linkage HAC.

\begin{proof}[Continued proof of \cref{lem:subunion-to-hac-upgma}]
  If we want to
  reduce incremental/decremental SubUnion to decremental/incremental
  average-linkage HAC, the update strategy in the proof of \cref{lem:subunion-to-hac-upgma}
  for the fully dynamic case 
  is invalid since we cannot add a weight-$w_y$ edge \emph{and} delete 
  a weight-$w_x$ edge (or vice versa).
  
  Instead, we construct the same HAC instance as in the fully dynamic
  case except that we make $y$ have edges to every $X_i \in X$
  rather than only $X_i \in S$. When processing updates, we skip
  modifying edges incident on $y$ and only add or remove the edges incident on
  $x$. Then we need to increase $\ell_x$ so that
  $C_x$ and $C_y$ do not merge with each other too early; we increase the constants as follows:
  \begin{align*}
    \ell_x &= 2\abs{X}w_y/(\lambda + 1) = O(\lambda^{8}\abs{U}^2\abs{X} + \lambda^5\abs{U}\abs{X}^2 + \lambda^2\abs{X}^3),  \\
    w_x &= \lambda(\ell_x + \abs{X})w_y + 1
  .\end{align*}
  The weight $w_x$ is chosen to be so large that all of the weight-$w_x$ edges
  merge before any weight-$w_y$ edge. After the weight-$w_x$ edges merge,
  the similarity between $C_x$ and $C_y$ is bounded above by 
  \begin{align*}
    \frac{\abs{X\setminus S}(w_y + \smallabs{C_y \cap U}L)}{\smallabs{C_x}\smallabs{C_y}}
    <
    \frac{\abs{X}(w_y + \smallabs{C_y \cap U}L)}{\ell_x\smallabs{C_y}} \\
    = \frac{\abs{X}w_y}{\ell_x\smallabs{C_y}}
    + \frac{\abs{X}\smallabs{C_y \cap U}L}{\ell_x\smallabs{C_y}}
    <
    \frac{\abs{X}w_y}{\ell_x}
    + \frac{\abs{X}L}{\ell_x}
    < \frac{2\abs{X}w_y}{\ell_x}
    = \lambda + 1,
  \end{align*}
  where the numerator of the first term comes from having a weight-$w_y$ edge from $C_y$ to 
  each $X_i$ in $C_x$ and from potentially having a weight-$L$ edge from every $X_i \in C_x$ to every $u \in U$
  contained in $C_y$. Comparing \cref{eq:subunion-to-hac-upgma:cy} to this upper bound, we find that
  all of the weight-$L$ edges incident on $u \in U$ covered by $S$ merge before
  $C_x$ and $C_y$ merge. Regardless of whether or not $C_x$ and $C_y$ merge,
  it is still true that $t$ either merges with an uncovered $u \in U$ if it
  exists or merges with $s$, and hence queries return the correct answer.
\end{proof}

Finally, we give reductions from SubUnion to complete-linkage HAC, weighted-average-linkage
HAC, and average-linkage \#HAC.  

\begin{figure}
  \includegraphics[scale=0.3]{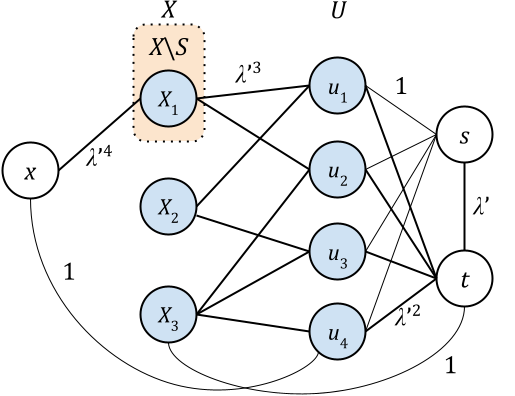}
  \caption{
  The figure illustrates the HAC instance constructed in 
  \cref{lem:subunion-to-hac-complete}'s reduction from SubUnion
  to complete-linkage HAC. 
  For cleanliness, we do not draw all the weight-$1$ edges from $x$ to each $u \in U$ and from
  $t$ to each $X_i \in X$.
  }
  \label{fig:subunion-to-hac-complete}
\end{figure}

\begin{lem}[Part of \cref{thm:hac-seth-lower-bounds}]\label{lem:subunion-to-hac-complete}
 Let $\lambda \in [1, \on{poly}(n)]$.
 Suppose we can solve dynamic / incremental / decremental $\lambda$-approximate
 complete-linkage HAC with
 $p(m,n)$ preprocessing work, $u(m,n)$ update work, and
 $q(m,n)$ query work. Then we can solve dynamic / decremental / incremental
 SubUnion with with $p(m', n')$ processing work, $u(m',n')$ update
 work, and $q(m',n')$ query work where $m' = O(\sum_i \abs{X_i})$ and $n' =
 O(\abs{U} + \abs{X})$.
\end{lem}
\begin{proof}
  Suppose we are given a SubUnion instance $(X, U)$ with subset $S \subseteq X$.
  Define  $\lambda' = \lambda +1$ and set the
  clustering threshold to be $\theta = \lambda'$.

  \emph{Preprocessing}: \Cref{fig:subunion-to-hac-complete} illustrates 
  the graph we will construct.
  Given a SubUnion instance $(X, U)$ with subset $S
  \subseteq X$, create a graph with a vertex representing each $X_i \in X$, a
  vertex representing each $u \in U$, and a weight-$\lambda'^3$ edge
  $\set{X_i, u}$ for each $u \in X_i$ for each $X_i \in X$. Add three extra
  vertices $x$, $s$, and $t$. Add a weight-$1$
  edge $\set{x, u}$ for each $u \in U$. Add a weight-$\lambda'^4$ edge $\set{x,
  X_i}$ for each $X_i \in X \setminus S$.  Add a weight-$\lambda'$ edge $\set{s,
  t}$. Add a weight-$1$ edge $\set{s, u}$ and a weight-$\lambda'^2$ edge
  $\set{t, u}$ for each $u \in U$. Add a weight-$1$ edge $\set{t, X_i}$ for each
  $X_i \in X$. The resulting graph $G$ has $n = O(\abs{X} + \abs{U})$ vertices and
  $m = O(\sum_i \abs{X_i})$ edges.  Initialize dynamic HAC on $G$
  with $\theta = \lambda'$.

  \emph{Update}: To simulate adding $X_i$ into $S$ in the SubUnion instance, delete
  edge ${x, X_i}$ in $G$. To remove $X_i$ from $S$, insert the
  weight-$\lambda'^4$ edge $\set{x, X_i}$.

  \emph{Query}: Query if vertices $s$ and $t$ are in the same cluster if we cluster up
  to similarity $\theta$. If yes, then return that $S$ covers $U$, otherwise
  return that $S$ does not cover $U$.

  \emph{Correctness}: All the weight-$\lambda'^4$ merge first, putting $x$ and all
  $X_i \in X \setminus S$ in the same cluster. Then all edges incident on this
  cluster have weight-$1$, so they don't participate in any more merges when
  clustering with threshold $\theta$.  Then the weight-$\lambda'^3$ edges
  incident on $X_i \in S$ merge, causing the subgraph induced by taking $X_i \in S$
  and all $u \in U$ to merge into its connected components. Connected components
  containing some $X_i$ will have a weight-$1$ edge to $t$, whereas the other
  connected components consisting of a singleton $u \in U$ will have a weight
  $\lambda'^2$ edge with $t$.

  Consider the case where $S$ covers $U$. In this case, all edges from $t$ to
  some $u \in U$ have weight $1$, and the next edge to merge is the
  weight-$\lambda'$ between $s$ and $t$. Hence $s$ and $t$ are in the same
  cluster, and a query will return the correct answer. In the case where $S$
  does not cover $U$, the next edge to merge is a weight-$\lambda'^2$ edge
  between $t$ and some $u \in U$. After that, $t$'s cluster has a weight-$1$
  edge with $s$, so $t$ does not merge with $s$ when clustering until threshold
  $\theta$. Hence a query will return the correct answer in this case too.
\end{proof}

\begin{figure}
  \includegraphics[scale=0.3]{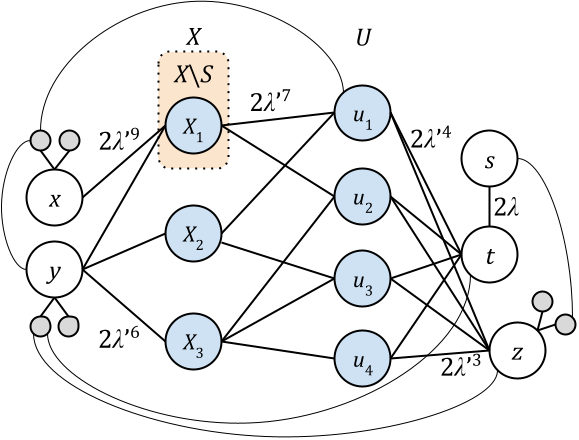}
  \caption{
  The figure illustrates the HAC instance constructed in 
  \cref{lem:subunion-to-hac-wpgma}'s reduction from SubUnion
  to weighted-average-linkage HAC. 
  In this example, the number of leaves (gray vertices) $\ell$ per star is $2$.
  Many weight-$1$ edges (the thinly drawn edges) from leaves to other vertices
  are omitted for cleanliness.
  }
  \label{fig:subunion-to-hac-wpgma}
\end{figure}

\begin{lem}[Part of \cref{thm:hac-seth-lower-bounds}]\label{lem:subunion-to-hac-wpgma}
  Let $\lambda \in [1, \on{poly}(n)]$.
 Suppose we can solve dynamic / incremental / decremental $\lambda$-approximate
 weighted-average-linkage HAC with
 $p(m,n)$ preprocessing work, $u(m,n)$ update work, and
 $q(m,n)$ query work. Then we can solve dynamic / decremental / incremental
 SubUnion with with $p(m', n')$ processing work, $u(m',n')$ update
 work, and $q(m',n')$ query work where $m' = O(\abs{U} \log \lambda + \sum_i \abs{X_i})$
 and $n' = O(\abs{U} + \abs{X} + \log \lambda)$.
\end{lem}
\begin{proof}
  Suppose we are given a SubUnion instance $(X, U)$ with subset $S \subseteq X$.
  Define $\lambda' = \lambda + 1$,
  define $\ell = \ceil{\log(2\lambda'^7)}$, and fix the clustering threshold to be
  $\theta = 2\lambda$. 

  \emph{Preprocessing}: 
  \Cref{fig:subunion-to-hac-wpgma} illustrates the graph we will construct.
  Create a graph
  $G$ with a vertex representing each $X_i \in X$, a vertex representing each $u
  \in U$, and a weight-$2\lambda'^7$ edge $\set{X_i, u}$ for each $u \in X_i$ for
  each $X_i \in X$. Add extra vertices $s$ and $t$.  Add a weight-$2\lambda$ edge
  $\set{s,t}$, and add a weight-$2\lambda'^4$ edge $\set{t,u}$ for each $u \in U$.
  We'll want to construct this graph so that $s$ and $t$ end up in the same
  cluster when clustering until threshold $\theta = 2\lambda$ if and only if $S$
  covers $U$.

  We'll add three star graphs to $G$ with centers $x$, $y$, and
  $z$, each with $\ell$ leaves. The purpose of the star centered on $x$ is to
  merge with $X_i \in X \setminus S$ and stop them from merging from with
  anything else. The purpose of star $z$ is that in the case where $S$ does not
  cover $U$, $t$ will merge with some uncovered $u \in U$, then merge with $z$,
  and finally the star will stop $t$ from merging with $s$. The purpose of star
  $y$ is to lower the weight from $u$ to $t$ for all $u \in U$ covered by $S$ so
  that if $S$ covers $u$, then $t$ will not merge with $z$.

  For the star graph with center $x$ and leaves $x_1, \ldots, x_\ell$, make the
  edge from $x$ to $x_i$ have weight $2\lambda'^8$ for each $x_i$. Connect $x$ to
  $X_i$ for each $X_i \in X \setminus S$ with weight $2\lambda'^9$. For each leaf $x_i$, connect $x_i$ to $y$ with
  weight $1$, and connect $x_i$ to $u \in U$ with weight $1$ for each $u \in U$.
  For the star graph with center $y$ and leaves $y_1, \ldots, y_\ell$, make the
  edge from $y$ to $y_i$ have weight $2\lambda'^5$ for each $y_i$. 
  Connect $y$ to $X_i$ for each $X_i \in X$ with weight $2\lambda'^6$.
  For each leaf
  $y_i$, add edges $\set{y_i, t}$ and $\set{y_i, z}$ with weight $1$. 
  For star graph with center $z$ and leaves $z_1, \ldots, z_\ell$, make the edge
  from $z$ to $z_i$ have weight $2\lambda'^2$ for each $z_i$, and connect each
  $z_i$ to $s$ with weight $1$. Connect $z$ to $u$ for each $u \in U$ with
  weight $2\lambda'^3$.

  \emph{Update and query}: Like
  in \cref{lem:subunion-to-hac-complete}, simulate adding or removing
  $X_i$ in $S$ by deleting or adding the weight-$2\lambda'^9$ weight edge
  $\set{x, X_i}$ in $G$. Answer queries about whether $S$ covers $U$ by returning
  whether $s$ and $t$ are in the same cluster if we run HAC until similarity
  $\theta$.

  \emph{Correctness}: Consider what happens when we run HAC up to similarity
  threshold $\theta = 2\lambda$. First, the weight-$2\lambda'^9$ edges merge so
  that each $X_i \in X \setminus S$ is in a cluster with $x$. Then the
  weight-$2\lambda'^8$ edges merge the leaves of $x$'s star with $x$ so that the cluster's remaining
  incident edges fall below weight $2$ (using the same inductive rationale as in the proof
  of \cref{lem:subconn-to-hac-wpgma}). Hence each $X_i \in X \setminus S$ no
  longer participates in any more merges. Then the weight-$2\lambda'^7$ edges
  incident on $X_i \in S$ merge, causing the subgraph induced by taking $X_i \in S$
  and all $u \in U$ to merge into its connected components. Each connected
  component containing an $X_i$ then merges with $y$ via its weight-$2\lambda'^6$ edge. The
  weight-$2\lambda'^5$ edges merge the leaves of $y$'s star with $y$ so that
  other edges incident on those connected components fall below $2$
  and don't participate in any more merges.

  In the case where $S$ covers $U$, the weight-$2\lambda'^2$ edges merge, and
  then the weight-$2\lambda$ edge merges, placing $s$ and $t$ in the same
  cluster. Hence a query will return the correct answer. In the case where $S$
  does not cover $U$, $t$ merges with each uncovered $u \in U$ via the
  weight-$2\lambda'^4$ edges, then merges with $z$ via a weight-$2\lambda'^3$
  edge. The weight-$2\lambda'^2$ edges merge the leaves of $z$'s star with $z$,
  decreasing the weight between $t$ and $s$ to below $2$. Hence $t$ and $s$ do
  not merge, and a query again returns the correct answer.
\end{proof}

\begin{figure}
  \includegraphics[scale=0.3]{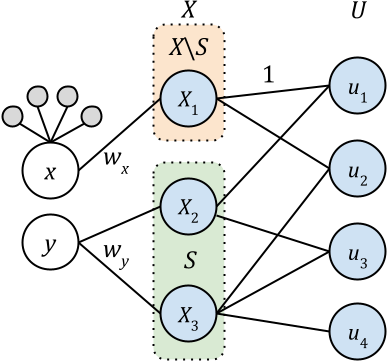}
  \caption{
  The figure illustrates the HAC instance constructed in 
  \cref{lem:subunion-to-hac-upgma-count}'s reduction from SubUnion
  to average-linkage \#HAC. 
  For cleanliness, we reduce the number of leaves (gray vertices) displayed attached to vertex $x$.
  }
  \label{fig:subunion-to-hac-upgma-count}
\end{figure}

\begin{lem}[Part of \cref{thm:hac-seth-lower-bounds}]\label{lem:subunion-to-hac-upgma-count}
  Let $\lambda \in [1, \on{poly}(n)]$.
   Suppose we can solve dynamic $\lambda$-approximate
   average-linkage \#HAC in
   $p(m,n)$ preprocessing work, $u(m,n)$ update work, and
   $q(m,n)$ query work. Then we can solve dynamic
   SubUnion with with $p(m', n')$ processing work, $u(m',n')$ update
   work, and $q(m',n')$ query work where $m'$ and $n'$ are
   $O(\lambda(\abs{X} + \abs{U})\abs{X})$.

   If we can solve incremental / decremental $\lambda$-approximate average-linkage
   \#HAC, then then the bounds hold for solve decremental / incremental
   SubUnion with $m'$ and $n'$ being $O(\lambda^2(\abs{X}+\abs{U})\abs{X}^2)$.
\end{lem}
\begin{proof}
  Suppose we are given a SubUnion instance $(X, U)$ with subset $S \subseteq X$.
  We start by describing the case where we have a fully dynamic algorithm rather
  than a partially dynamic algorithm.
  Define the following constants:
  \begin{align*}
    \theta &= 1/(\abs{X}+\abs{U}), \\
    w_y &= \lambda\abs{X} + 1 = O(\lambda\abs{X}), \\
    \ell_x &= \lambda\abs{X}/\theta = O(\lambda(\abs{X} + \abs{U})\abs{X}), \\
    w_x &= \lambda(\ell_x + \abs{X}) + 1 = O(\lambda^2(\abs{X} + \abs{U})\abs{X})
  .\end{align*}

  \emph{Preprocessing}: \Cref{fig:subunion-to-hac-upgma-count} illustrates the graph 
  we will construct.
  Create a graph $G$ with a vertex representing each $X_i
  \in X$, a vertex representing each $u \in U$, and a weight-$1$ edge $\set{X_i,
  u}$ for each $u \in X_i$ for each $X_i \in X$. Add an extra vertex $y$, and
  add a weight-$w_y$ edge from $y$ to $X_i$ for each $X_i \in S$. Add a star
  graph with center $x$ and $\ell_x$ leaves connected to the center with weight
  $w_x$. Add a weight-$w_x$ edge from $x$ to $X_i$ for each $X_i \in X\setminus
  S$. This graph has $O(\ell_x)$ vertices and edges. Initialize \#HAC on
  this graph.

  \emph{Update}: Simulate adding $X_i$ to $S$ by adding a weight-$w_y$ edge from $X_i$
  to $y$ and removing the weight-$w_x$ from $X_i$ to $x$. Similarly, simulate removing
  $X_i$ from $S$ by removing edge $\set{X_i, y}$ and adding edge $\set{X_i, x}$.

  \emph{Query}: Query whether there are exactly two clusters when performing HAC
  up to threshold $\theta$. If yes, then return that $S$ covers $U$, otherwise
  return that $S$ does not cover $U$.

  \emph{Correctness}: The weights $w_x$ and $w_y$ are chosen to be large enough
  that all the edges of weight $w_x$ and $w_y$ merge before any weight-$1$ edge. Let
  $C_x$ denote $x$'s cluster (containing $X \setminus S$) and $C_y$ denote $y$'s
  cluster (containing $S$). From this point onwards, the similarity between $C_x$ to
  another cluster $C$ is always at most
  \begin{align*}
    \frac{\abs{X\setminus S}\abs{C \cap U}}{\abs{C_x}\abs{C}}
    \le
    \frac{\abs{X\setminus S}}{\abs{C_x}}
    <
    \frac{\abs{X}}{\ell_x}
    = \frac{\theta}{\lambda}
  ,\end{align*}
  so $C_x$ experiences no more merges. On the other hand, consider $C_y$, which
  has edges connecting it to every $u \in U$ covered by $S$. Each $u \in U$ is a
  singleton cluster until it merges with $C_y$. The similarity between $C_y$ and
  some $u \in U$ covered by $S$ and not yet merged with $C_y$ is always at least
  \begin{align*}
    \frac{1}{\abs{X} + \abs{U}} = \theta
  ,\end{align*}
  so every $u \in U$ covered by $S$ merges with $C_y$ when clustering until
  threshold $\theta$.

  After all the merges, in the case where $S$ covers $U$, we return the correct
  answer to a query since the only two clusters are $C_x$ and $C_y$. Otherwise,
  if $S$ does not cover $U$, then there will be more than two clusters since
  each uncovered $u \in U$ will be in a singleton cluster. Hence we answer
  queries correctly in both cases.

  \emph{Partially dynamic}: The update strategy is invalid if we want to
  reduce incremental/decremental SubUnion to decremental/incremental \#HAC.
  Instead, construct the graph so that $y$ has edges to every $X_i
  \in X$ rather than only $X_i \in S$. When processing updates, we skip
  modifying edges incident on $y$ and only add or remove the edges incident on
  $x$. Then we need to increase $\ell_x$ so that
  $C_x$ and $C_y$ do not merge with each other; we adjust the constants as follows:
  \begin{align*}
    \ell_x &= 2\lambda\abs{X}w_y/\theta = O(\lambda^2(\abs{X}+\abs{U})\abs{X}^2), \\
    w_x &= \lambda(\ell_x + \abs{X})w_y + 1 = O(\lambda^4(\abs{X} + \abs{U})\abs{X}^3)
  .\end{align*}
  The weight $w_x$ is chosen to be so large that all the weight-$w_x$ edges
  merge before the weight-$w_y$ edges. The last change in the correctness
  analysis to make sure that $C_y$ and $C_x$ never merge from this point onwards.
  The weight between $C_x$ and $C_y$ is bounded above by
  \begin{align*}
    \frac{\abs{X\setminus S}(w_y + \smallabs{C_y \cap U})}{\smallabs{C_x}\smallabs{C_y}}
    < \frac{\abs{X}(w_y + \smallabs{C_y \cap U})}{\ell_x\smallabs{C_y}} \\
    = \frac{\abs{X}w_y}{\ell_x\smallabs{C_y}}
    + \frac{\abs{X}\smallabs{C_y \cap U}}{\ell_x\smallabs{C_y}}
    < \frac{\abs{X}w_y}{\ell_x}
    + \frac{\abs{X}}{\ell_x}
    < \frac{2\abs{X}w_y}{\ell_x}
    = \frac{\theta}{\lambda}
  ,\end{align*}
  so indeed $C_y$ and $C_x$ do not merge.
\end{proof}

Existing lower bounds conditional on SETH for
SubUnion apply directly to complete-linkage HAC and weighted-average-linkage HAC, whereas to
turn \cref{lem:subunion-to-hac-upgma-count} into the lower bounds in
\cref{thm:hac-seth-lower-bounds}, we apply \cref{lem:subunion-to-approx}.

}

\end{document}